\documentclass[11pt,letterpaper]{article}
\usepackage{times,mysetup,my-acm}
\usepackage{tweaklist}
\FULLPAGE

\newcommand{\Bobby}[1]{\sidenote{Bobby}{#1}}

\newcommand{\payoff}{{\pi}}
\newcommand{\distrib}{{\mathcal{D}}}
\newcommand{\prob}{{\mathbb{P}}}
\newcommand{\expect}{{\mathbf{E}}}

\renewcommand{\enumhook}{\setlength{\topsep}{2pt}%
  \setlength{\itemsep}{0pt} }
\renewcommand{\itemhook}{\setlength{\topsep}{2pt}%
  \setlength{\itemsep}{0pt} }

\renewcommand{\ASnote}[1]{}
\renewcommand{\Bobby}[1]{}

\newcommand{\Cov}{\ensuremath{\mathtt{Cov}}}
\newcommand{\LCD}{\ensuremath{\mathtt{LCD}}}
\newcommand{\MaxMinLCD}{\ensuremath{\mathtt{MaxMinLCD}}}

\newcommand{\FFproblem}{full-feedback Lipschitz experts problem}
\newcommand{\ULproblem}{uniformly Lipschitz experts problem}

\providecommand{\Appendix}{}
\renewcommand{\Appendix}[2][?]{%
	\refstepcounter{section}%
	\vspace{\parskip}%
	{\flushright\Large\bfseries\appendixname\ \thesection: #1}%
	\vspace{\baselineskip}%
}

\renewcommand{\appendix}{%
	\renewcommand{\section}{\secdef\Appendix\Appendix}%
	\renewcommand{\thesection}{\Alph{section}}%
	\setcounter{section}{0}%
}

\title{Sharp Dichotomies for Regret Minimization in Metric Spaces%
\footnote{This is the full version of a paper that will appear in~\emph{ACM-SIAM SODA}, 2010.}}

\author{
{\Large Robert Kleinberg}\thanks{Computer Science Department, Cornell University, Ithaca, NY 14853. Email:~{\tt rdk at cs.cornell.edu}.  Supported by NSF awards
CCF-0643934 and IIS-0905467, an Air Force OFfice of Scientific Research grant, a Microsoft Research New Faculty Fellowship, and an
 Alfred P. Sloan Foundation Fellowship.}
\and
{\Large Aleksandrs Slivkins}\thanks{
Microsoft Research, Mountain View, CA 94043.
Email:~{\tt slivkins at microsoft.com}.}
}

\date{December 2008 \\ Revised: April 2009, October 2009}

\begin{document}

\maketitle

\newcommand{\given}{\,|\,}

\hyphenation{interval perhaps finite recent results imply}

\begin{abstract}
The Lipschitz multi-armed bandit (MAB) problem generalizes the classical multi-armed bandit problem by assuming one is given side information consisting of a priori upper bounds on the difference in expected payoff between certain pairs of strategies.  Classical results of Lai-Robbins~\cite{Lai-Robbins-85} and Auer et al.~\cite{bandits-ucb1} imply a logarithmic regret bound for the Lipschitz MAB problem on finite metric spaces.  Recent results on continuum-armed bandit problems and their generalizations imply lower bounds of $\sqrt{t}$, or stronger, for many infinite metric spaces such as the unit interval.  Is this dichotomy universal?  We prove that the answer is yes: for every metric space, the optimal regret of a Lipschitz MAB algorithm is either bounded above by any $f\in \omega(\log t)$, or bounded below by any $g\in o(\sqrt{t})$. Perhaps surprisingly, this dichotomy does not coincide with the distinction between finite and infinite metric spaces; instead it depends on whether the completion of the metric space is compact and countable.  Our proof connects upper and lower bound techniques in online learning with classical topological notions  such as perfect sets and the Cantor-Bendixson theorem.

We also consider the \emph{full-feedback} (a.k.a., \emph{best-expert}) version of Lipschitz MAB problem, termed the \emph{Lipschitz experts problem}, and show that this problem exhibits a similar dichotomy. We proceed to give nearly matching upper and lower bounds on regret in the Lipschitz experts problem on uncountable metric spaces. These bounds are of the form $\tilde{\Theta}(t^\gamma)$, where the exponent $\gamma\in [\tfrac12, 1]$ depends on the metric space. To characterize this dependence, we introduce a novel dimensionality notion tailored to the experts problem. Finally, we show that both Lipschitz bandits and Lipschitz experts problems become completely intractable (in the sense that no algorithm has regret $o(t)$) if and only if the completion of the metric space is non-compact.
\end{abstract}


{\bf ACM Categories and subject descriptors:}

{\bf General Terms:} theory, algorithms.

\newpage

\section{Introduction}

Multi-armed bandit (henceforth, MAB) problems have been studied for more than fifty years as a clean abstract setting for analyzing the exploration-exploitation tradeoffs that are common in sequential decision making.  In the \emph{stochastic MAB problem}, an algorithm must repeatedly choose from a fixed set of strategies (a.k.a. ``arms"), each time receiving a random payoff whose distribution depends on the strategy selected.\footnote{More precisely, the payoff of each arm is an independent sample from a fixed time-invariant distribution with bounded support.} The performance of MAB algorithms is commonly evaluated in terms of \emph{regret}: the difference in expected payoff between the algorithm's choices and always playing one fixed strategy.  In addition to their many applications --- which range from experimental design to online auctions and web advertising --- another appealing feature of multi-armed bandit algorithms is that they are surprisingly efficient in terms of the growth rate of regret: for finite-armed bandit problems, algorithms whose regret at time $t$ scales as $O(\log t)$ have been known for more than two decades, beginning with the seminal work of Lai and Robbins~\cite{Lai-Robbins-85} and extended in subsequent work such as~\cite{bandits-ucb1}.

Many of the applications of MAB problems --- especially the computer science applications such as online auctions, web advertising, or adaptive routing --- require considering strategy sets which are very large or even infinite.  For infinite strategy sets the $O(\log t)$ bound does not apply, while for very large finite sets the $O(\cdot)$ notation masks a prohibitively large constant.  Indeed, without making any assumptions about the strategies and their payoffs, bandit problems with large strategy sets allow for no non-trivial solutions --- any MAB algorithm performs as badly, on some inputs, as random guessing. This motivates the study of bandit problems in which the strategy set is large but one is given \emph{side information} constraining the form of the payoffs. Such problems have  become the subject of quite intensive study in recent years, e.g.~\cite{sundaram-bandits-92, Agrawal-bandits-95, bandits-exp3, Bobby-nips04, Bobby-thesis, McMahan-colt04, Bobby-stoc04,FlaxmanKM-soda05, Hayes-soda06, Cope/06, AuerOS/07, DaniHK-nips07, KakadeKL/07, Hazan-colt07, Hazan-soda09}.

The \emph{Lipschitz MAB problem} is a version of the stochastic MAB problem in which the side information consists of \emph{a priori} upper bounds on the difference in expected payoff between certain pairs of strategies. This models situations where the decision maker has access to some similarity information about strategies which ensures that similar strategies obtain similar payoffs.  Abstractly, the similarity information may be modeled as defining a \emph{metric space} structure on the strategy set, and the side constraints imply that the expected payoff function $\mu$ is a Lipschitz function (with Lipschitz constant $1$) on this metric space. 

The Lipschitz MAB problem was introduced by Kleinberg et al.~\cite{LipschitzMAB-stoc08}.\footnote{Megiddo and Hazan~\cite{Hazan-colt07} consider a somewhat related (but technically very different) setting which combines full feedback, contextual ``hints", convex payoffs, and (essentially) a similarity metric space on the contexts.}  Preceding work~\cite{Agrawal-bandits-95,AuerOS/07,Cope/06,Bobby-nips04,yahoo-bandits07} has
studied the problem  in a few specific metric spaces such as a one-dimensional real interval. The prior work considered regret $R(t)$ as a function of time $t$, and focused on the asymptotic dependence of $R(t)$ on, loosely speaking, the dimensionality of the metric space. Various upper and lower bounds of the form $R(t) = \tilde{\Theta}(t^\gamma)$ were proved, where the exponent $\gamma<1$ depends on the metric space. In particular, if the metric space is the interval $[0,1]$ with the standard metric $d(x,y) = |x-y|$, then there exists an algorithm with regret $R(t) = \tilde{O}(t^{2/3})$, and this bound is tight up to polylog factors~\cite{Bobby-nips04}. More generally, for an arbitrary infinite metric space $(X,d)$ one can define an isometry invariant $\gamma = \gamma(X,d) \in [\tfrac12, 1]$ such that there exists an algorithm with regret $R(t) = \tilde{O}(t^\gamma)$, which is tight up to polylog factors if $\gamma>\tfrac12$; see~\cite{LipschitzMAB-stoc08}.

The following picture emerges.  Although algorithms with regret
$R(t) =  O(t^{\gamma}), \, \gamma < 1$ are known for most metric
spaces, existing work unfortunately provides no examples of
infinite metric spaces admitting bandit algorithms satisfying
the Lai-Robbins regret bound $R(t) = O(\log t)$, although this
bound holds for
all finite metrics.  In fact, for most metric spaces that have
been studied (such as the unit interval) this possibility is
\emph{excluded} by known lower bounds of the form $R(t) \not\in
o(t^{\gamma}),$ where $\gamma \geq \frac12.$
Therefore it is natural to ask,
\newpage
\begin{itemize}
\item Is   $\tilde{O}(\sqrt{t})$ regret the best possible for an infinite metric space? Alternatively, are there infinite metric spaces for which one can achieve regret $O(\log t)$? Is there any metric space for which the best possible regret is \emph{between} $O(\log t)$ and $\tilde{O}(\sqrt{t})$?
\end{itemize}

\xhdr{Our contributions.}
To make the above issue more concrete, let us put forward the following definition.
\begin{definition}\label{def:tractability}
Consider the Lipschitz MAB problem on a fixed metric space. A bandit algorithm is \emph{$f(t)$-tractable} if for any problem instance $\mathcal{I}$ the algorithm's regret is
	$R(t) = O_{\mathcal{I}}(f(t))$.~\footnote{The notation $O_{\mathcal{I}}()$ means that the constant in $O()$ can depend on $\mathcal{I}$.}
The problem is \emph{$f(t)$-tractable} if such an algorithm exists.
\end{definition}
We settle the questions listed above by proving the following dichotomy.

\begin{theorem}\label{thm:main-MAB}
Consider the Lipschitz MAB problem on a fixed metric space $(X,d)$. Then the following dichotomy holds: either the problem is $f(t)$-tractable for every $f\in \omega(\log t)$, or it is not $g(t)$-tractable for any $g\in o(\sqrt{t})$.  In fact, the former occurs if and only if the completion of $X$ is a compact metric space with countably many points.
\end{theorem}

It is worth mentioning that the regret bound $R(t) = O_{\mathcal{I}}(\log t)$ is the best possible, even for two-armed bandit problems, by a lower bound of Lai and Robbins~\cite{Lai-Robbins-85}.  Thus our upper bound for Lipschitz MAB problems in compact, countable metric spaces is nearly the best possible bound for such spaces, modulo the gap between ``$f(t) = \log t$" and ``$\forall f\in\omega(\log t)$". Furthermore, we show that this gap is inevitable for infinite metric spaces:

\begin{theorem}\label{thm:logT}
For every infinite metric space $(X,d)$, the Lipschitz MAB problem on $(X,d)$ is not $(\log t)$-tractable.
\end{theorem}

We turn our attention to the \emph{full-feedback} version of the Lipschitz MAB problem. For any MAB problem there exists a corresponding \emph{full-feedback} problem in which after each round, the payoffs from all strategies are revealed.\footnote{Formally, an algorithm can query an arbitrary finite number of strategies.} Such settings have been extensively studied in the online learning literature under the name \emph{best experts problems}~\cite{experts-jacm97,CesaBL-book,vovk98}. In particular, for a finite set of strategies one can achieve a \emph{constant} regret~\cite{sleeping-colt08} when payoffs are i.i.d. over time.

In addition to the full feedback, one could also consider a version in which the payoffs are revealed for some but not all strategies. Specifically, we define the \emph{double feedback}, where in each round the algorithm selects two strategies: the ``bet" for which it receives the payoff, and the ``free peek". After the round, the payoffs are revealed for both strategies. By abuse of notation, we will treat the bandit setting as a special case of the experts setting.

The experts version of the Lipschitz MAB problem, called the \emph{Lipschitz experts problem}, is defined in the obvious way: a problem instance is specified by a triple $(X,d,\prob)$, where $(X,d)$ is a metric space and $\prob$ is a Borel probability measure on the set $[0,1]^X$ of payoff functions on $X$ (with the Borel $\sigma$-algebra
induced by the product topology on $[0,1]^X$)
such that the expected payoff function
    $x \mapsto E_{f \in \prob}[f(x)]$
is a Lipschitz function on $(X,d)$. In each round an algorithm is presented with an i.i.d. sample from $\prob$. The metric structure of $(X,d)$ is known to the algorithm, the measure $\prob$ is not.  We show that the Lipschitz experts problem exhibits a dichotomy similar to the one in Theorem~\ref{thm:main-MAB}. We formulate the upper bound for the double feedback, and the lower bound for the full feedback, thus avoiding the issue of what it means for an algorithm to receive feedback for infinitely many strategies.

\begin{theorem}\label{thm:main-experts}
The Lipschitz experts problem on a fixed metric space $(X,d)$ is either $1$-tractable, even with double feedback, or it is not $g(t)$-tractable for any $g\in o(\sqrt{t})$, even with full feedback. The former occurs if and only if the completion of $X$ is a compact metric space with countably many points.
\end{theorem}

Theorems~\ref{thm:main-MAB} and~\ref{thm:main-experts} assert a dichotomy between metric spaces on which the Lipschitz MAB/experts problem is very tractable, and those on which it is somewhat tractable. Let us consider the opposite end of the ``tractability spectrum" and ask for which metric spaces the problem becomes completely intractable. We obtain a precise characterization: the problem is completely intractable if and only if the metric space is not pre-compact. Moreover, our upper bound is for the bandit setting, whereas the lower bound is for full feedback.

\begin{theorem}\label{thm:boundary-of-tractability}
The Lipschitz experts problem on a fixed metric space $(X,d)$ is either $f(t)$-tractable for some $f\in o(t)$, even in the bandit setting, or it is not $g(t)$-tractable for any $g\in o(t)$, even with full feedback. The former occurs if and only if the completion of $X$ is a compact metric space.
\end{theorem}

Consider the \FFproblem. In view of the $\sqrt{t}$ lower bound from Theorems~\ref{thm:main-experts}, we are interested in matching upper bounds. Gupta et al.~\cite{Anupam-experts07} observed that such bounds hold for every metric space $(X,d)$ of finite covering dimension: namely, the Lipschitz experts problem on $(X,d)$ is $\sqrt{t}$-tractable. (Their algorithm is a version of the ``naive algorithm'' from~\cite{Bobby-nips04, LipschitzMAB-stoc08}.) Therefore it is natural to ask whether there exist metric spaces for which the optimal regret in the Lipschitz experts problem is \emph{between} $\sqrt{t}$ and $t$. We settle this question by proving a characterization with nearly matching upper and lower bounds in terms of a novel dimensionality notion tailored to the experts problem.

\begin{theorem}\label{thm:experts-MaxMinLCD}
For any metric space $(X,d)$, there exist an isometry invariant $b= b(X,d)$ such that the \FFproblem\ on $(X,d)$ is $(t^\gamma)$-tractable for any $\gamma> \tfrac{b+1}{b+2}$, and not $(t^\gamma)$-tractable for any $\gamma < \tfrac{b-1}{b}$. Depending on the metric space, $b(X,d)$ can take any value on a dense subset of
    $[0, \infty)$.
\end{theorem}

The lower bound in Theorem~\ref{thm:experts-MaxMinLCD} holds for a restricted version of \FFproblem\ in which a problem instance $(X,d,\prob)$ satisfies a further property that each function $f\in \mathtt{support}(\mathbb{P})$ is itself a Lipschitz function on $(X,d)$. We term this version the \emph{\ULproblem} (with full feedback). In fact, for this version we obtain a matching upper bound.

\begin{theorem}\label{thm:experts-MaxMinLCD-unif}
Consider the \ULproblem\ with full feedback. Fix an uncountable metric space $(X,d)$. Let $b= b(X,d)$ the isometry invariant from Theorem~\ref{thm:experts-MaxMinLCD}. Then the problem on $(X,d)$ is $(t^\gamma)$-tractable for any
    $\gamma> \max(\tfrac{b-1}{b}, \tfrac12)$,
and not $(t^\gamma)$-tractable for any
    $\gamma < \max(\tfrac{b-1}{b}, \tfrac12)$.
\end{theorem}

\OMIT{ 
To this end, we consider ``very high-dimensional" metric spaces such as exponentially branching edge-weighted trees and the space of all probability distributions on $[0,1]^d$ under the Earthmover distance. We introduce a novel dimensionality notion which better captures the complexity of such spaces and characterizes the regret of the ``naive" experts algorithm from~\cite{Anupam-experts07}. (Interestingly, the \ULproblem{} allows for a very non-trivial improvement in the analysis.)
} 

\OMIT{ 
For any metric space $(X,d)$, there exist an isometry-invariant parameters
$\gamma_*(X,d)$ and $\gamma^*(X,d)$ such that the \FFproblem\ on $(X,d)$ is $(t^\gamma)$-tractable for any $\gamma>\gamma^*(X,d)$, and not $(t^\gamma)$-tractable for any $\gamma<\gamma_*(X,d)$. Depending on the metric space, $\gamma^*(X,d)$ can take any value in $[\tfrac12, 1)$.  There exist metric spaces for which
$\tfrac12 < \gamma_*(X,d) < \gamma^*(X,d) < 1.$
} 

\OMIT{For any metric space $(X,d)$, there exists an isometry invariant $\gamma = \gamma(X,d)$ such that the full-feedback Lipschitz experts problem on $(X,d)$ is $f(t)$-tractable for any $f\in \omega(t^\gamma)$, and not $g(t)$-tractable for any $g\in o(t^\gamma)$.}

\OMIT{
\begin{theorem}\label{thm:experts-dim}
Consider the Lipschitz experts problem on a fixed metric space $(X,d)$. Then either the problem is $\sqrt{t}$-tractable, even with single feedback, or it is not $\sqrt{t}$-tractable, even with full feedback. The former occurs if and only if the max-min-covering dimension is finite.
\end{theorem}
} 


\OMIT{
We introduce several new techniques, the most important of which appear in the (joint) proof of the two main results -- Theorem~\ref{thm:main-MAB} and Theorem~\ref{thm:main-experts}.
} 

\xhdr{Connection to point-set topology.}
The main technical contribution of this paper is an interplay of online learning and point-set topology, which requires novel algorithmic and lower-bounding techniques. In particular, the connection to topology is essential in the (joint) proof of the two main results (Theorem~\ref{thm:main-MAB} and Theorem~\ref{thm:main-experts}). There, we identify a simple topological property (\emph{well-orderability}) which entails the algorithmic result, and another topological property (\emph{perfectness}) which entails the lower bound.

\begin{definition}\label{def:topology}
Consider a topological space $X$.
$X$ is called \emph{perfect} if it contains no isolated points.
A \emph{topological well-ordering} of $X$ is a well-ordering
$(X,\prec)$ such that every initial segment thereof is an open
set.  If such $\prec$ exists, $X$ is called \emph{well-orderable}.
A metric space $(X,d)$ is called well-orderable if and only if
its metric topology is well-orderable.
\end{definition}

Perfect spaces are a classical notion in point-set topology.
Topological well-orderings are implicit in
the work of Cantor~\cite{Cantor83}, but the particular definition
given here is new, to the best of our knowledge.

The proof of Theorems~\ref{thm:main-MAB} and~\ref{thm:main-experts} (for compact metric spaces) consists of three parts: the algorithmic result for a compact well-orderable metric space, the lower bound for a metric space with a perfect subspace, and the following lemma that ties together the two topological properties.

\begin{lemma}\label{lm:topological-equivalence}
For any compact metric space $(X,d)$, the following are equivalent: (i) $X$ is a countable set, (ii) $(X,d)$ is well-orderable, (iii) no subspace of $(X,d)$ is perfect.\footnote{For arbitrary metric spaces we have (ii)$\iff$(iii) and (i)$\Rightarrow$(ii), but not (ii)$\Rightarrow$(i). }
\end{lemma}
Lemma~\ref{lm:topological-equivalence} follows from classical theorems of Cantor-Bendixson~\cite{Cantor83} and Mazurkiewicz-Sierpinski~\cite{MazSier}.  We provide a proof in Appendix~\ref{sec:topological}
for the sake of making our exposition self-contained.

To reduce the Lipschitz MAB problem to complete metric spaces we show that the problem is $f(t)$-tractable on a given metric space if and only if it is $f(t)$-tractable on the completion thereof. Same is true for the double-feedback Lipschitz experts problem, and the ``only if" direction holds for the \FFproblem. Then the main dichotomy results follow from the lower bound in Theorem~\ref{thm:boundary-of-tractability}.

\xhdr{Accessing the metric space.}
We define a bandit algorithm as a (possibly randomized) Borel
measurable function that
maps a history of past observations $(x_i, r_i) \in X\times [0,1]$ to
a strategy $x\in X$ to be played in the current period.
An experts algorithm is similarly defined as a (possibly
randomized) Borel measurable function mapping the observation
history to a strategy $x \in X$ to be played in the current
period.  (Or, in the case of the double feedback model,
a pair of strategies representing the ``bet'' and ``free peek''.)
The observation history is either a sequence
of elements of $[0,1]^X$ in the full feedback model,
or a sequence of quadruples
         $(x_i, r_i, x'_i, r'_i) \in (X \times [0,1])^2$
in the double feedback model.

These definitions abstract away a potentially thorny issue of representing
and accessing an infinite metric space. For our algorithmic results,
we handle this issue as follows: the metric space is accessed via
well-defined calls to a suitable \emph{oracle}.
Moreover, the main algorithmic result in Theorems~\ref{thm:main-MAB} and~\ref{thm:main-experts} requires an oracle which represents the well-ordering. We also provide an extension in Section~\ref{sec:simpler-alg}: an $\omega(\log t)$-tractability result for a wide family of metric spaces -- including, for example, compact metric spaces with a finite number of limit points -- for which a more intuitive oracle access suffices. These are the metric spaces with a finite \emph{Cantor-Bendixson rank}, a classic notion from point-set topology.

\xhdr{Related work and discussion.}
Algorithms for the stochastic MAB problem admit regret guarantees of the form
	$R(t) = O(f(t))$,
which are of two types -- \emph{instance-specific} and \emph{instance-independent} -- depending on whether the constant in $O()$ is allowed to depend on the problem instance. For instance, {\sc ucb1}~\cite{bandits-ucb1} admits an instance-specific guarantee $R(t) = O(\log t)$, whereas the best-known instance-independent guarantee for this algorithm is only $R(t) = O(\sqrt{kt \log t})$, where $k$ is the number of arms. Accordingly, a lower bound for the instance-independent version has to show that for any algorithm and a given time $t$, there exists a problem instance whose regret is large at this time, whereas for the instance-specific version one needs a much more ambitious argument: for any algorithm there exists a problem instance whose regret is large \emph{infinitely often}. In this paper, we focus on instance-specific guarantees.

\OMIT{ 
Apart from the stochastic MAB problem considered in this paper, several other  formulations have been studied in the literature on multi-armed bandits, which spans Operations Research, Economics and Computer Science (see~\cite{CesaBL-book} for background).\footnote{We only mention the major MAB formulations. A proper survey of the numerous extensions is beyond the scope of this paper.}
} 

Apart from the stochastic MAB problem considered in this paper, several other MAB formulations have been studied in the literature (see~\cite{CesaBL-book} for background). Early work~\cite{Gittins-index,Gittins-book} has focused on Bayesian formulations in which Bayesian priors on payoffs are known, and goal is to maximize the payoff in expectation over these priors. In these formulations, an MAB instance is a \emph{Markov Decision Process} (MDP) in which each arm is represented by a Markov Chain with rewards on states, and the transition happens whenever the arm is played.  In the more ``difficult" \emph{restless bandits} \cite{Whittle-88, BertsimasMora-2000, NinoMora-01} formulations, the state also changes when the arm is passive, according to another transition matrix. In the theoretical computer science literature, recent work in this vein includes~\cite{GuhaMunagala-focs07, GuhaMunagala-soda09}. Interestingly, these Bayesian formulations have an offline flavor: given the MDP, one needs to efficiently compute a (nearly) optimal mapping from states to actions. Contrasting the Bayesian formulations in which the probabilistic model is fully specified, the \emph{adversarial MAB problem}~\cite{bandits-exp3,Auer-focs00,Hazan-soda09} makes no stochastic assumptions whatsoever. Instead, it makes a very pessimistic assumption that payoffs are chosen by an adversary that has access to the algorithm's code but not to its random seed. As in the stochastic MAB problem, the goal is to minimize regret. For any fixed (finite) number of arms, the best possible regret in this setting is $R(t) = O(\sqrt{t})$~\cite{bandits-exp3}. For infinite strategy sets, one often considers the \emph{linear MAB problem} in which strategies lie in a convex subset of $\R^d$, and in each round the payoffs form a linear function~\cite{McMahan-colt04, Bobby-stoc04, DaniHK-nips07, AbernethyHR-colt08, Hazan-soda09}.

\OMIT{, or more generally, a convex function~\cite{Bobby-nips04, FlaxmanKM-soda05}.}

It is an open question whether the ideas from the Lipschitz MAB problem extend to the above formulations. The adversarial version of the Lipschitz MAB problem is well-defined, but to the best of our knowledge, the only known result is the ``naive" algorithm from~\cite{Bobby-nips04}. One could define the stochastic version of the linear MAB problem (in which the expected payoffs form a fixed time-invariant linear function), which can be viewed as a special case of the Lipschitz MAB problem. However, this view is not likely to be fruitful because in the Lipschitz MAB problem measuring a payoff of one arm is useless for estimating the payoffs of distant arms, whereas in prior work on the linear MAB problem inferences about distant arms are crucial. For Bayesian MAB problems with limited similarity information, it is not clear how to model this information, mainly because in the Bayesian setting similarity between arms is naturally represented via correlated priors rather than a metric space.

\OMIT{An \emph{oblivious adversary} must fix all payoffs in advance before the first round; an \emph{adaptive adversary} sees the choices made by the algorithm in all previous rounds.}

\xhdr{Organization of the paper.}
Preliminaries are in Section~\ref{sec:prelims}. We present a joint proof for the two main results (Theorems~\ref{thm:main-MAB} and~\ref{thm:main-experts}).  The lower bound is proved in Section~\ref{sec:lower-bound} and the algorithmic results are in Section~\ref{sec:tractability}. Coupled with the topological equivalence (Lemma~\ref{lm:topological-equivalence}), this gives the proof for compact metric spaces. A complementary $(\log t)$-intractability result for infinite metric spaces (Theorem~\ref{thm:logT}) is in Section~\ref{sec:logT}.  The $\omega(\log t)$-tractability result via simpler oracle access (for metric spaces of finite Cantor-Bendixson rank) is in Section~\ref{sec:simpler-alg}.  The boundary-of-tractability result  (Theorems~\ref{thm:boundary-of-tractability}) is in Section~\ref{sec:boundary-body}. The \FFproblem\ in a (very) high dimension (including Theorems~\ref{thm:experts-MaxMinLCD} and~\ref{thm:experts-MaxMinLCD-unif}) is discussed in Sections~\ref{sec:FFproblem} and~\ref{sec:FFproblem-characterization}.

Some of the proofs are moved to appendices. In Appendix~\ref{sec:reduction} we reduce the problem to that on complete metric spaces.  All KL-divergence arguments (which underlie our lower bounds) are gathered in Appendix~\ref{sec:KL-divergence}.
We provide a self-contained proof of the topological lemma  (Lemma~\ref{lm:topological-equivalence}) in Appendix~\ref{sec:topological}.

\section{Preliminaries.}
\label{sec:prelims}

This section contains various definitions which make the paper essentially self-contained (the only exception being \emph{ordinal numbers} which are used in Section~\ref{sec:MaxMinLCD-UB}). In particular, the paper uses notions from General Topology which are typically covered in any introductory text or course on the subject.

 \xhdr{Lipschitz MAB problem.} Consider the Lipschitz MAB problem on a metric space $(X,d)$ with payoff function $\mu$. The payoff from each arm $x\in X$ is an independent sample from a fixed (time-invariant) distribution with support in $[0,1]$ and expectation $\mu(x)$ such that $|\mu(x)-\mu(y)| \leq d(x,y)$ for all $x,y\in X$. For $S\subset X$ denote
	$\sup(\mu,S) = \sup_{x\in S} \mu(x)$
and similarly
	$\argmax(\mu,S) = \argmax_{x\in S} \mu(x)$.
Given a bandit algorithm \A, let $P_{(\A,\mu)}(t)$ be the expected reward collected by the algorithm in the first $t$ rounds on the problem instance $(X,d,\mu)$. The \emph{regret} of algorithm \A\ in $t$ rounds is
	$R_{(\A,\mu)}(t) = \sup(\mu, X)\,t -  P_{(\A,\mu)}(t)$.
Given a Lipschitz experts algorithm \A\ and a problem instance
$(X,d,\prob)$, the notations $P_{(\A, \, \prob)}(t)$ and
$R_{(\A, \, \prob)}(t)$ --- denoting expected reward and
regret --- are defined analogously.

\xhdr{Metric topology.} Let $(X,d)$ be a metric space. An open ball in $(X,d)$ is denoted
	$B(x_0,r) = \{x\in X:\, d(x,x_0) <r\}$,
where $x_0\in X$ is the center, and $r\geq 0$ is the radius.
A \emph{Cauchy sequence} in $(X,d)$
is a sequence such that for every $\delta>0$,
there is an open ball of radius $\delta$ containing all
but finitely many points of the sequence.  We say
$X$ is \emph{complete} if every Cauchy sequence has a
limit point in $X$.  For two Cauchy sequences $\mathbf{x}=x_1,x_2,\ldots$
and $\mathbf{y}=y_1,y_2,\ldots$ the \emph{distance}
$d(\mathbf{x},\mathbf{y}) = \lim_{i \rightarrow \infty} d(x_i,y_i)$
is well-defined.  Two Cauchy sequences are declared to be
equivalent if their distance is $0$.  The equivalence
classes of Cauchy sequences form a metric space $(X^*,d)$
called the \emph{completion} of $(X,d)$. The subspace of all constant sequences
is identified with $(X,d)$: formally, it is a dense subspace of $(X^*,d)$
which is isometric to $(X,d)$.
A metric space $(X,d)$ is \emph{compact} if every collection
of open balls  covering $(X,d)$ has a finite subcollection
that also covers $(X,d)$.  Every compact metric space
is complete, but not vice-versa.

Let $X$ be a set. A family $\F$ of subsets of $X$ is called a \emph{topology} if it contains $\emptyset$ and $X$ and is closed under arbitrary unions and finite intersections. When a specific topology is fixed and clear from the context, the elements of $\F$ are called \emph{open sets}, and their complements are called \emph{closed sets}.  Throughout this paper, these terms will refer to the \emph{metric topology} of the underlying metric space, the smallest topology that contains all open balls (namely, the intersection of all such topologies). A point $x$ is called \emph{isolated} if the singleton set $\{x\}$ is open. A function between topological spaces is \emph{continuous} if the inverse image of every open set is open.

\xhdr{Set theory.} Let $S$ be a set. A \emph{well-ordering} on a set $S$ is a total order on $S$ with the property that every non-empty subset of $S$ has a least element in this order. Each set can be well-ordered. (More precisely, this statement is equivalent to the Axiom of Choice.)

In Section~\ref{sec:MaxMinLCD-UB} use \emph{ordinals}, a.k.a. \emph{ordinal numbers}, are a classical concept in set theory that, in some sense, extend natural numbers beyond infinity. Understanding this paper requires only the basic notions about ordinals, namely the standard (von Neumann) definition of ordinals, successor and limit ordinals, and transfinite induction. The necessary material can be found in any introductory text on Mathematical Logic and Set Theory, and also on \emph{Wikipedia}.

\section{Lower bounds via a perfect subspace}
\label{sec:lower-bound}

In this section we prove the following lower bound:

\begin{theorem}\label{thm:lower-bound}
Consider the Lipschitz experts problem on a metric space $(X,d)$ which has a perfect subspace. Then the problem is not $g$-tractable for any $g\in o(\sqrt{t})$. In fact, a much stronger result holds: there exist a distribution $\mathcal{P}$ over problem instances $\mu$ such that for any experts algorithm \A\ we have
\begin{align}\label{eq:lower-bound}
 (\forall g\in o(\sqrt{t}))\quad
\Pr_{\mu\in\mathcal{P}}
	\left[ R_{(\A,\,\mu)}(t) = O_{\mu}(g(t)) \right] = 0.
\end{align}
\end{theorem}

Let us construct the desired distribution over problem instances. First, we use the existence of a perfect subspace to construct a useful system of balls.

\begin{definition}
A \emph{ball-tree} on a metric space $(X,d)$ is a complete infinite binary tree whose nodes are pairs $(x, r)$, where $x\in X$ is the ``center'' and $r\in (0,1]$ is the ``radius'', such that:
\begin{OneLiners}
\item if $(x,r)$ is a parent of $(x',r')$ then $d(x,x') + r' < r/2$,
\item if $(x, r_x)$ and $(y,r_y)$ are siblings, then
	$r_x + r_y < d(x,y)$.
\end{OneLiners}
\end{definition}

In a ball-tree, each tree node $(x,r)$ corresponds to a ball $B(x,r)$ so that each child is a subset of its parent and any two siblings are disjoint.\footnote{Defining internal nodes as \emph{balls} rather than $(x,r)$ pairs could lead to confusion later in the construction because a ball in a metric space a set of points, and as such does not necessarily have a unique center or radius.}

\begin{lemma}\label{lm:ball-tree}
For any metric space with a perfect subspace there exists a ball-tree.
\end{lemma}

\begin{proof}
Consider a metric space $(X,d)$ with a perfect subspace $(Y,d)$. Let us construct the ball-tree recursively, maintaining the invariant that for each tree node $(y,r)$ we have $y\in Y$. Pick an arbitrary $y\in Y$ and let the root be $(y,1)$. Suppose we have constructed a tree node $(y,r)$, $y\in Y$.  Since $Y$ is perfect, the ball $B(y,r/4)$ contains another point  $y'\in Y$. Let $r' = d(y,y')/2$ and define the two children of $(y,r)$ as $(y,r')$ and $B(y',r')$.
\end{proof}

Now let us use the ball-tree to construct the distribution on payoff functions. Consider a metric space $(X,d)$ with a fixed ball-tree $T$.  For each $i\geq 1$, let $D_i$ be the set of all depth-$i$ tree nodes, and let
	$r^*_i = \min \{r: (x,r)\in D_i   \}$
be the smallest radius among these nodes. Note that $r^*_i\leq 2^{-i}$.
Choose a number $n_i$ large enough that $g(n) < \tfrac{1}{8i} r^*_i \sqrt{n}$
for all $n > n_i$,
and let $\delta_i = n_i^{-1/2}$.
For each tree node $w = (x_0, r_0)$ 
define a function
	$F_{w}: X \rightarrow [0,1]$ by
\begin{align}\label{eq:needle}
 F_{w}(x) = \begin{cases}
	\min\{ r_0 - d(x,x_0),\, r_0/2  \} & \text{if $x\in B(x_0, r_0)$}, \\
	0	& \text{otherwise.}
\end{cases}
\end{align}

It is easy to see that $F_{w}$ is a Lipschitz function on $(X,d)$.
A \emph{leaf} in a ball-tree is an infinite path from the root:
	$\mathbf{w} = (w_0, w_1, w_2,\,\ldots)$, where $ w\in D_i$ for all $i$.
A \emph{lineage} in a ball-tree is a
set of tree nodes containing at most
one child of each node; if it contains
\emph{exactly} one child of each node
then we call it a \emph{complete lineage}.
For each complete lineage
$\lambda$ there is an associated leaf $\mathbf{w}(\lambda)$
defined by $\mathbf{w}=(w_0,w_1,\, \ldots)$ where
$w_0$ is the root and for $i>0$, $w_i$ is the unique
child of $w_{i-1}$ that belongs to $\lambda$.
Let us use a lineage in the ball-tree
to define a probability measure $\prob_{\lambda}$
on payoff functions via the following sampling rule.  First every
node $w$ independently
samples a random sign $\mathrm{sign}(w) \in \{+1,-1\}$,
assigning probability $(1+\delta_i)/2$ to $+1$ if $w \in \lambda \cap D_i$,
and choosing the sign
uniformly at random otherwise.  Now define a payoff function
$\payoff$ associated with this sign pattern as follows:
        $ \payoff = \tfrac12 + \sum_{w \in T \setminus D_0}
\mathrm{sign}(w) F_{w}.$
By construction, $\payoff$ is a Lipschitz function taking
values in $[0,1]$.
Let $\mathcal{P}_{T}$ be the distribution over problem instances
$\prob_{\lambda}$ in which  $\lambda$ is a complete lineage
sampled uniformly at random; that is, each node samples one of its children
independently and uniformly at random, and $\lambda$ is the set of
sampled children. This completes our construction.

\begin{note}{Remark.}
Let $\lambda$ be a complete lineage in the ball-tree, and let
$\mathbf{w}(\lambda) = (w_0, w_1, w_2,\,\ldots)$ be its associated leaf,
where $w_i = (x_i,r_i)$ for all $i$.
Suppose for some $i$ we have
	$x\in B(x_i, r_i/2)$ and $x\not\in B(x_{i+1}, r_{i+1})$.
Then the expected payoff function $\mu_{\lambda} = \expect[\pi]$
associated to the measure $\prob_{\lambda}$ satisfies
	$\mu_{\lambda}(x) = \tfrac12 + \textstyle{\sum_{j=1}^i r^*_i \delta_i/4}$.
\end{note}

\begin{lemma}\label{lm:ball-tree-LB}
Consider a metric space $(X,d)$ with a ball-tree $T$. Then~\refeq{eq:lower-bound} holds with $\mathcal{P} = \mathcal{P}_T$.
\end{lemma}

\OMIT{ 
Consider a set $X$ and two functions $\mu_1, \mu_{-1}: X\rightarrow [0,1]$. These functions are called \emph{\eps-alternatives}, $\eps>0$, if for some disjoint subsets $S_1, S_{-1}\subset X$, some $\mu_0 \in [\tfrac13, \tfrac23]$ and each  $i\in \{1, -1\} $ the following holds:
(i) $\mu_1 \equiv \mu_{-1} \leq \mu_0$ on $X \setminus (S_1 \cup S_2)$,
(ii) $\mu_i \equiv \mu_0 $ on $S_{-i}$, and
(iii) $\sup_{x\in S_i} \mu_i(x) = \mu_0+\eps$.
} 

\OMIT{ 
To prove this lemma, we use the lower-bounding techniques of~\cite{bandits-exp3,Bobby-nips04,LipschitzMAB-stoc08}. In fact, we present a framework that clarifies and generalizes these techniques, in a way that makes them applicable to the Lipschitz experts problem. To this end, we consider a more general setting than the one in the Lipschitz experts problem. In the latter, the underlying metric space $(X,d)$ implicitly defines the set \F\ of all feasible payoff functions. In  \emph{the feasible experts problem}, instead of the metric space one is given an arbitrary \F. A problem instance thus consists of a triple $(X,\F,\mu)$, where $X$ and \F\ are revealed to the algorithm, and $\mu\in\F$ is not.
}

\OMIT{ 
We use the lower-bounding technique for the basic $k$-armed bandit problem, due to Auer et al.~\cite{bandits-exp3}. Here we apply this technique to the experts problem. For a cleaner exposition, we encapsulate the usage of this technique in a theorem. The setting in this theorem is considerably  more general than the one in the original lower bound in~\cite{bandits-exp3}, although the underlying ideas are similar. For the bandits version, a similar extension is implicitly used, but not explicitly formulated, in~\cite{LipschitzMAB-stoc08}.\ASnote{can we put a more positive spin on this?}

It is more natural to formulate this result in a more general setting than the Lipschitz experts problem.
} 

To prove this lemma, we define a notion called an
$(\eps,\delta,k)$-ensemble, which is a collection of
$k$ payoff distributions that are nearly
indistinguishable from the standpoint of an
online learning algorithm.  To this end, we
consider a more general setting than the one in
the Lipschitz experts problem.
In the \emph{feasible experts problem}, one
is given a set $X$ (not necessarily a metric
space) along with a collection $\distrib$ of Borel
probability measures on the set $[0,1]^X$ of
functions $\payoff : X \rightarrow [0,1].$  A problem
instance of the feasible experts problem
consists of a triple $(X,\distrib,\prob)$
where $X$ and $\distrib$ are known to the
algorithm, and $\prob \in \distrib$ is not.

\OMIT{ 
In the following definition,
$\vec{\prob} = (\prob_0,\prob_1,\ldots,\prob_k)$
is a $(k+1)$-tuple of Borel probability measures on $[0,1]^X$,
$\payoff \,:\, X \rightarrow [0,1]$ denotes a random sample
from $[0,1]^X$ under any of these probability measures,
and $\F$ denotes the Borel $\sigma$-field of $[0,1]^X$.
Finally, for $0 \leq i \leq k$ and $x \in X$,
$\mu_i(x)$ denotes the expectation of $\payoff(x)$
under measure $\prob_i$.
}

\begin{definition}\label{def:ensemble}
Consider a set $X$ and a $(k+1)$-tuple
$\vec{\prob} = (\prob_0,\prob_1,\ldots,\prob_k)$
of Borel probability measures on $[0,1]^X$, the
set of $[0,1]$-valued payoff functions $\payoff$
on $X$.  For $0 \leq i \leq k$ and $x \in X$, let
$\mu_i(x)$ denote the expectation of $\payoff(x)$
under measure $\prob_i$.
We say that $\vec{\prob}$ is an \emph{$(\eps,\delta,k)$-ensemble}
if there exist pairwise disjoint subsets $S_1,S_2,\ldots,S_k \subseteq X$
for which the following properties hold:
\begin{enumerate}
\item \label{ens:1}
for every $i$ and every event $\mathcal{E}$ in the Borel
$\sigma$-algebra of $[0,1]^X$, we have
    $1-\delta < \prob_0(\mathcal{E}) / \prob_i(\mathcal{E}) < 1+\delta,$
\item \label{ens:2}
for every $i > 0$, we have
    $\sup(\mu_i, S_i) - \sup(\mu_i,\, X \setminus S_i) \geq \eps.$
\end{enumerate}
\end{definition}

\OMIT{ 
\begin{definition}
 An \emph{$(\eps,k)$-ensemble} on $(X,\F)$ is
 a collection of subsets $\F_1, \ldots, \F_k \subset \F$ such that there exist mutually disjoint subsets
        $S_1, \ldots, S_k \subset X$ and a number $\mu_0 \in [\tfrac13, \tfrac23
]$
for which the following properties hold. Let
        $S=\cup_{i=1}^k S_i$
and pick any $\mu_i\in \F_i$ for each $i$. Then for each $i$ we have:
(i)  $\mu_i \equiv \mu_1 \leq \mu_0$ on $X \setminus S$,
(ii) $\mu_i \equiv \mu_0 $ on $S\setminus S_i$, and
(iii)  $\sup(\mu_i, S_i) = \mu_0+\eps$.
\end{definition}
}

\begin{theorem}\label{thm:LB-technique}
Consider the feasible experts problem on $(X,\distrib)$. Let  $\vec{\prob}$ be an $(\eps,\delta,k)$-ensemble with $\{\prob_1,\ldots,\prob_k\} \subseteq
\distrib$ and $0<\eps,\delta<1/2$. Then for any
    $t < \ln(17k)/(2 \delta^2)$
and any experts algorithm \A, at least half of the measures $\prob_i$ have the property that 
	$R_{(\A,\,\prob_i)}(t) \geq \eps t/2$.
\end{theorem}

\OMIT{ 
\begin{theorem}\label{thm:LB-technique}
Let  $\F_1, \ldots, \F_k$ be an $(\eps,k)$-ensemble on $(X,\F)$, where $0<\eps<1/12$.
\begin{itemize}
\item[(a)] Consider the feasible MAB problem on $(X,\F)$.
Then for any $t \leq k/32\,\eps^2$ and any bandit algorithm \A\ there exists $\F_i$ such that for each payoff function $\mu_i\in \F_i$ we have
	$R_{(\A,\,\mu_i)}(t) \geq \eps t/60$.
\item[(b)] Consider the feasible experts problem on $(X,\F)$. Then for any
$t \leq \ln(k)/32\,\eps^2$
and any experts algorithm \A\ there exists $\F_i$ such that for each payoff function $\mu_i\in \F_i$ we have
	$R_{(\A,\,\mu_i)}(t) \geq \eps t/60$.
\end{itemize}
\end{theorem}
} 

\begin{note}{Remarks.}
For space reasons, the proof of this theorem has been moved to the
appendix.
The proof of Theorem~\ref{thm:lower-bound} uses Theorem~\ref{thm:LB-technique} for $k=2$, and  the proof of Theorem~\ref{thm:experts-MaxMinLCD}  will use it again for large $k$.
\end{note}

\begin{proofof}{Lemma~\ref{lm:ball-tree-LB}}
Consider the ball-tree $T$. For each $i\geq 1$, recall that $D_i$
is the set of all depth-$i$ nodes in $T$, and that
	$r^*_i = \min \{r: (x,r)\in D_i   \}$
is the smallest radius among these nodes. Let $\mathbf{P}$ be the set of all probability measures induced by the lineages of $T$.  For each complete lineage
$\lambda$ and tree node $w$ in $T$, let $w_1,w_2$ denote
the children of $w$,  let $i$ denote their depth, and
let $w'$ denote the unique element of
$\{w_1,w_2\} \cap \lambda$.  The three lineages
    $\lambda_0 = \lambda \setminus \{w'\}, \,
     \lambda_1 = \lambda_0 \cup \{w_1\}, \,
     \lambda_2 = \lambda_0 \cup \{w_2\}$
define a triple of probability measures
    $\vec{\prob} = (\prob_{\lambda_0}, \prob_{\lambda_1}, \prob_{\lambda_2})$
that constitute a $(\eps,\delta_i,2)$-ensemble where
$\eps = r^*_i \delta_i / 4$.

Let us fix an experts algorithm \A.  By Theorem~\ref{thm:LB-technique}
there exists $\alpha(w) \in \{\prob_{\lambda_1},\prob_{\lambda_2}\}$
such that for any $t_i$ satisfying $1/\delta_i^2 < t < \ln(34)/(2 \delta_i^2)$,
	$$R_{(\A,\, \alpha(w))}(t_i) \geq \eps t_i / 2
           = r^*_i \delta_i t_i / 8 > \tfrac18 r^*_i \sqrt{t},$$
Recalling the definition of $n_i = \delta_i^{-2}$, we see that
        $i \cdot g(t_i) < \tfrac{1}{8} r^*_i \sqrt{t_i} <
R_{(\A, \, \alpha(w))}(t_i).$

For each $i$, let us define $\mathcal{E}_i$ to be the set of
input distributions $\prob_{\lambda}$ such that $\lambda$ is
a complete lineage whose associated leaf
$\mathbf{w}(\lambda) = (w_0,w_1,\ldots)$ satisfies
$w_i = \alpha(w_{i-1})$.  Interpreting these sets as random
events under the probability distribution $\mathcal{P}_T$,
we have proved the following: there exists a sequence of
events $\mathcal{E}_i$, $i\in \N$ and a sequence of times
$t_i \rightarrow \infty$ such that for each $i$ we have
(i) $\Pr[\mathcal{E}_i |\,
	\sigma(\mathcal{E}_1,\, \ldots,\, \mathcal{E}_{i-1})] = \tfrac12$
and (ii)	
	$R_{(\A,\,\prob)}(t_i) > i \cdot g(t_i)$
for any $\prob \in \mathcal{E}_i$.

Now, let us fix an experts algorithm \A. 
For each complete lineage  $\lambda$, define
$ C_\lambda := \inf \{ C \leq \infty:\,
        R_{(\A,\,\prob_{\lambda})}(t) \leq C\, g(t) \text{~for all $t$} \} $.
Note that
        $R_{(\A,\,\prob_{\lambda})}(t) = O_\mu(g(t)) $
if and only if $C_\lambda < \infty$. We claim that
    $\Pr [C_\lambda < \infty] = 0$
where the probability is over the random choice of complete
lineage $\lambda$.  Indeed,
if infinitely many events $\mathcal{E}_i$ happen,
then event $\{C_\mu<C\}$ does not.  But the
probability that infinitely many events $\mathcal{E}_i$
happen is 1, because for every positive integer $n$,
   $\Pr \left[ \cap_{i=n}^{\infty}
      \overline{\mathcal{E}_i} \right] =
     \prod_{i=n}^{\infty} \Pr \left[ \overline{\mathcal{E}_i}
     \,\left|\, \cap_{j=n}^{i-1} \overline{\mathcal{E}_j} \right. \right]
     = 0.$
\end{proofof}

\section{Tractability for compact well-orderable metric spaces}
\label{sec:tractability}

\newcommand{\UCB}{\ensuremath{\text{{\sc ucb1}}}}

In this section we prove the main algorithmic result.

\begin{theorem}\label{thm:main-alg}
Consider a compact well-orderable metric space $(X,d)$. Then:
\begin{OneLiners}
\item[(a)] the Lipschitz MAB problem on $(X,d)$ is $f$-tractable for every $f\in\omega(\log t)$;
\item[(b)] the Lipschitz experts problem on $(X,d)$ is 1-tractable, even with a double feedback.
\end{OneLiners}
\end{theorem}

We present a joint exposition for both the bandit and the experts version. Let us consider the Lipschitz MAB/experts problem on a compact metric space $(X,d)$ with a topological well-ordering $\prec$ and a payoff function $\mu$. For each strategy $x\in X$, let
	$S(x) = \{y\preceq x: y\in X\}$
be the corresponding initial segment of the well-ordering $(X,\prec)$. Let
	$\mu^* = \sup(\mu, X)$
denote the maximal payoff. Call a strategy $x\in X$ \emph{optimal} if $\mu(x) = \mu^*$.
We rely on the following structural lemma:

\begin{lemma}\label{lm:structural}
There exists an optimal strategy $x^*\in X$ such that
	$\sup(\mu, X\setminus S(x^*)) < \mu^*$.
\end{lemma}

\begin{proof}
Let $X^*$ be the set of all optimal strategies. Since $\mu$ is a continuous real-valued function on a compact space $X$, it attains its maximum, i.e. $X^*$ is non-empty, and furthermore $X^*$ is closed. Note that $\{S(x): x\in X^*\}$ is an open cover for $X^*$. Since $X^*$ is compact (as a closed subset of a compact set) this cover contains a finite subcover, call it $\{S(x): x\in Y^*\}$. Then the $\prec$-maximal element of $Y^*$  is the $\prec$-maximal element of $X^*$. The initial segment $S(x^*)$ is open, so its complement $ Y = X\setminus S(x^*)$ is closed and therefore compact. It follows that $\mu$ attains its maximum on $Y$, say at a point $y^*\in Y$. By the choice of $y^*$ we have $x^*\prec y^* $, so by the choice of $x^*$ we have $\mu(x^*)> \mu(y^*)$.
\end{proof}

In the rest of this section we let $x^*$ be the strategy from Lemma~\ref{lm:structural}. Our algorithm is geared towards finding $x^*$ eventually, and playing it from then on. The idea is that if we cover $X$ with balls of a sufficiently small radius, any strategy in a ball containing  $x^*$ has a significantly larger payoff than any strategy in a ball that overlaps with $X\setminus S(x^*)$.

The algorithm accesses the metric space and the well-ordering via the following two oracles.

\begin{definition}\label{def:covering-oracle}
A \emph{$\delta$-covering} of a metric space $(X,d)$ is a subset $S\subset X$ such that each point in $X$ lies within distance $\delta$ from some point in $S$. An oracle $\mathcal{O} = \mathcal{O}(k)$ is a \emph{covering oracle} for $(X,d)$ if it inputs $k\in\N$ and outputs a pair $(\delta, S)$ where $\delta = \delta_\mathcal{O}(k)$ is a positive number and  $S$ is a $\delta$-covering of $X$ consisting of at most $k$ points. Here $\delta_\mathcal{O}(\cdot)$ is any function such that
	$\delta_\mathcal{O}(k)\rightarrow 0$ as $k\rightarrow\infty$.
\end{definition}

\begin{definition}\label{def:ordering-oracle}
Given a metric space $(X,d)$ and a total order $(X,\prec)$, the \emph{ordering oracle} inputs a finite collection of balls (given by the centers and the radii), and returns the $\prec$-maximal element covered by the closure of these balls, if such element exists, and an arbitrary point in $X$ otherwise.
\end{definition}

\newcommand{\oracleX}{\ensuremath{x_{\mathrm{or}}}}
\newcommand{\muAv}{\ensuremath{\mu_{\mathrm{av}}}}
\newcommand{\explSub}{\ensuremath{\mathtt{EXPL}}}

Our algorithm is based on the following \emph{exploration subroutine} $\explSub()$.

\begin{algorithm}\label{alg:PMO-explore}
Subroutine $\explSub(k,n,r)$: inputs $k, n \in \N$ and $r\in (0,1)$, outputs a point in $X$.

First it calls the covering oracle $\mathcal{O}(k)$ and receives a $\delta$-covering $S$ of $X$ consisting of at most $k$ points. Then it plays each strategy $x\in S$ exactly $n$ times; let $\muAv(x)$ be the sample average. Let us say that $x$ a \emph{loser} if
	$\muAv(y) - \muAv(x)> 2r  + \delta$
for some $y\in S$. Finally, it calls the ordering oracle with the collection of all closed balls 	 $\Bar{B}(x,\delta)$ such that $x$ is not a loser, and outputs the point $\oracleX\in X$ returned by this oracle call.
\end{algorithm}

Clearly, $\explSub(k,n,r)$ takes at most $kn$ rounds to complete. We show that for sufficiently large $k,n$ and sufficiently small $r$ it returns $x^*$ with high probability.

\begin{lemma}\label{lm:tractability}
Fix a problem instance and let $x^*$ be the optimal strategy from Lemma~\ref{lm:structural}.
Consider increasing functions $k,n,T: \N\to \N$ such that
    $r(t) := 4\sqrt{ (\log T(t))\, / n(t)}  \to 0$.
Then for any sufficiently large $t$, with probability at least $1-T^{-2}(t)$, the subroutine
    $\explSub(k(t),\,n(t),\, r(t))$ returns $x^*$.
\end{lemma}
\begin{proof}
Let us use the notation from Algorithm~\ref{alg:PMO-explore}. Fix $t$ and consider a run of
    $\explSub(k(t),\,n(t),\, r(t))$.
Call this run \emph{clean} if for each $x\in S$ we have $|\muAv(x)-\mu(x)| \leq r(t)$. By Chernoff Bounds, this happens with probability at least $1-T^{-2}(t)$. In the rest of the proof, let us assume that the run is clean.

Let $\Bar{B}$ be the union of the closed balls $\Bar{B}(x,\delta)$, $x\in S^*$. Then the ordering oracle returns the $\prec\mbox{-maximal}$ point in $\Bar{B}$ if such point exists. We will show that
	$x^*\in \Bar{B}\subset S(x^*)$
for any sufficiently large $t$, which will imply the lemma.

We claim that $x^* \in \Bar{B}$. Since $S$ is a $\delta$-covering, there exists $y^*\in S$ such that
	$d(x^*, y^*)\leq \delta$.
Let us fix one such $y^*$. It suffices to prove that $y^*$ is not a loser. Indeed, if 	
	$\muAv(y) - \muAv(y^*)> 2\,r(t) + \delta $
for some $y\in S$ then
	$\mu(y) > \mu(y^*) + \delta \geq \mu^*$,
contradiction. Claim proved.

Let
	$\mu_0 = \sup(\mu, X\setminus S(x^*))$
and let $r_0 = (\mu^*-\mu_0)/7$. Let us assume that $t$ is sufficiently large so that $r(t) < r_0$ and
	$\delta = \delta_{\mathcal{O}}(k(t))<r_0$,
where $\delta_{\mathcal{O}}(\cdot)$ is from the definition of the covering oracle.

We claim that $\Bar{B}\subset S(x^*)$. Indeed, consider
	$x\in S$ and $y\in X\setminus S(x^*)$ such that $d(x, y) \leq \delta$.
It suffices to prove that $x$ is a loser. Consider some $y^*\in S$ such that $d(x^*,y^*)\leq \delta$. Then by the Lipschitz condition
\begin{align*}
\muAv(y^*)
	& \geq \mu(y^*) - r_0 \geq \mu^* -  2 r_0, \\
\muAv(x)
	&\leq \mu(x) + r_0 \leq \mu(y) + r_0
	\leq \mu_0 + 2r_0 \leq \mu^* - 5 r_0 \\
\muAv(y^*) - \muAv(x)
	& \geq 3 r_0 > 2 r(t) + \delta. \qedhere
\end{align*}
\end{proof}

\begin{proofof}{Theorem~\ref{thm:main-alg}}
Let us fix a function $f\in \omega(\log t)$. Then $f(t) = \alpha(t) \log(t)$ where $\alpha(t)\to\infty$. Without loss of generality, assume that $\alpha(t)$ is non-decreasing. (If not, then instead of $f(t)$ use
    $g(t) = \beta(t) \log(t)$,
where
    $\beta(t) = \inf \{ \alpha(t'):\, t'\geq t \}$.)

For part (a), define
    $k_t = \flr{\sqrt{g(t)/ \log t}} $,
    $n_t = \flr{k_t \log t}$,
and
    $r_t = 4\sqrt{(\log t)/ n_t}$.
Note that $r_t \to 0$.

The algorithm proceeds in phases of a doubly exponential length\footnote{The doubly exponential phase length is necessary in order to get $f$-tractability.  If we employed the more familiar \emph{doubling trick} of using phase length $2^i$ (as in~\cite{bandits-exp3,Bobby-nips04,LipschitzMAB-stoc08} for example) then the algorithm would only be $f(t)\, \log t$-tractable.}. A given phase  $i=1,2,3,\ldots$ lasts for $T = 2^{2^i}$ rounds. In this phase, first we call the exploration subroutine
    $\explSub(k_T,\,n_T,\, r_T)$.
Let $\oracleX\in X$ be the point returned by this subroutine. Then we play \oracleX\ till the end of the phase. This completes the description of the algorithm.

Fix a problem instance $\mathcal{I}$. Let $W_i$ be the total reward accumulated by the algorithm in phase $i$, and let
	$R_i = 2^{2^i}\,\mu^* - W_i$
be the corresponding share of regret.
By Lemma~\ref{lm:tractability} there exists $i_0 = i_0(\mathcal{I})$ such that for any phase $i\geq i_0$  we have, letting $T=2^{2^i}$ be the phase duration, that $R_i \leq k_T\, n_T \leq g(T)$ with probability at least $1-T^{-2}$, and therefore
	$E[R_i] \leq g(T) + T^{-1}$.
For any $t > t_0 = 2^{2^{i_0}}$ it follows by summing over
$i \in \{i_0,i_0+1,\ldots,\lceil \log \log t \rceil\}$
that
        $R_{\A,\,\mathcal{I}}(t) = O(t_0 + g(t)).$
Note that we have used the fact that $\alpha(t)$ is non-decreasing.

For part (b), we separate exploration and exploitation. For exploration, we run  $\explSub()$ on the \emph{free peeks}. For exploitation, we use the point returned by $\explSub()$ in the previous phase. Specifically, define
    $k_t = n_t = \flr{\sqrt{t}} $,
and
    $r_t = 4\sqrt{(t^{1/4})/ n_t}$.
The algorithm proceeds in phases of exponential length. A given phase $i=1,2,3,\ldots$ lasts for $T = 2^i$ rounds. In this phase, we run the exploration subroutine
    $\explSub(k_T,\,n_T,\, r_T)$
on the \emph{free peeks}. In each round, we \emph{bet} on the point returned by $\explSub()$ in the previous phase. This completes the description of the algorithm.

By Lemma~\ref{lm:tractability} there exists $i_0 = i_0(\mathcal{I})$ such that in any phase $i\geq i_0$ the algorithm incurs zero regret with probability at least $1-e^{\Omega(i)}$. Thus the total regret after $t>2^{i_0}$ rounds is at most $t_0 + O(1)$.
\end{proofof}

\section{The $(\log t)$-intractability for infinite metric spaces: proof of Theorem~\ref{thm:logT}}
\label{sec:logT}

Consider an infinite metric space $(X,d)$. In view of Theorem~\ref{thm:boundary-of-tractability}, we can assume that the completion $X^*$ of $X$ is compact. It follows that there exists $x^*\in X^*$ such that $x_i\to x^*$ for some sequence $x_1, x_2,\, \ldots\, \in X$. Let $r_i = d(x_i, x^*)$. Without loss of generality, assume that
    $r_{i+1} < \tfrac12\, r_i$
for each $i$, and that the diameter of $X$ is $1$.

Let us define an ensemble of payoff functions
    $\mu_i :X\to[0,1]$, $i\in\N$,
where $\mu_0$ is the ``baseline" function, and for each $i\geq 1$ function $\mu_i$ is the ``counterexample" in which a neighborhood of $x_i$ has slightly higher payoffs.
The ``baseline" is defined by
    $\mu_0(x) = \tfrac12 - \tfrac{d(x,x^*)}{8}$,
and the ``counterexamples" are given by
$$ \mu_i(x) = \mu_0(x)+ \nu_i(x),
\text{~~where~~}
 \nu_i(x) = \tfrac{3}{4} \max \left(0, \tfrac{r_i}{3} - d(x,x^*) \right).
$$
Note that both $\mu_0$ and $\nu_i$ are $\tfrac18$-Lipschitz
and $\tfrac34$-Lipschitz w.r.t. $(X,d)$, respectively,
so $\mu_i$ is $\tfrac78$-Lipschitz w.r.t $(X,d)$.
Let us fix a MAB algorithm \A\ and assume that it is $(\log t)$-tractable. Then for each $i\geq 0$ there exists a constant $C_i$ such that
    $R_{(\A,\, \mu_i)}(t) < C_i \log t$
for all times $t$. We will show that this is not possible.

Intuitively, the ability of an algorithm to distinguish between payoff functions $\mu_0$ and $\mu_i$, $i\geq 1$  depends on the number of samples in the ball
    $B_i = B(x_i,\, r_i/3)$.
(This is because $\mu_0 = \mu_i$ outside $B_i$.) In particular, the number of samples itself cannot be too different under $\mu_0$ and under $\mu_i$, \emph{unless it is large}. To formalize this idea, let  $N_i(t)$ be the number of times algorithm \A\ selects a strategy in the ball $B_i$ during the first $t$ rounds, and let $\sigma(N_i(t))$ be the corresponding $\sigma$-algebra. Let $\prob_i[\cdot]$ and $\mathbb{E}_i[\cdot]$ be, respectively, the distribution and expectation induced by $\mu_i$. Then we can connect $\mathbb{E}_0[N_i(t)]$ with the probability of any event $S\in \sigma(N_i(t))$ as follows.

\begin{claim}\label{cl:logT-KLdiv}
For any $i\geq 1$ and any event $S\in \sigma(N_i(t))$ it is the case that
\begin{align}\label{eq:logT-KLdiv}
 \prob_i[S] < \tfrac13 \leq \prob_0[S] \quad \Rightarrow \quad
 - \ln(\prob_i[S]) - \tfrac{3}{e}
    \leq O(r_i^2)\; \mathbb{E}_0[N_i(t)].
\end{align}
\end{claim}

\noindent
Claim~\ref{cl:logT-KLdiv} is proved using KL-divergence techniques, see Appendix~\ref{sec:KL-divergence} for details. To complete the proof of the theorem, we claim that for each $i\geq 1$ it is the case that
    $\mathbb{E}_0 [N_i(t)] \geq \Omega(r_i^{-2}\, \log t)$
for any sufficiently large $t$. Indeed, fix $i$ and let
    $S = \{N_i(t) < r_i^{-2} \log t\}$.
Since
    $$C_i \log t > R_{(\A,\; \mu_i)}(t) \geq
        \prob_i(S)\,(t - r_i^{-2} \log t) \tfrac{r_i}{8} ,$$
it follows that
    $\prob_i(S) < t^{-1/2} < \tfrac13$
for any sufficiently large $t$. Then by Claim~\ref{cl:logT-KLdiv} either
    $\prob_0(S) < \tfrac13$
or the consequent in~\refeq{eq:logT-KLdiv} holds. In both cases
        $\mathbb{E}_0 [N_i(t)] \geq \Omega(r_i^{-2}\, \log t)$.
Claim proved.

Finally, the fact that $\mu_0(x^*) - \mu_0(x) \geq r_i/12$
for every $x \in B_i$ implies that
     $R_{(\A, \, \mu_0)}(t) \geq \tfrac{r_i}{12} \mathbb{E}_0[N_i(t)]
      \geq \Omega(r_i^{-1} \, \log t)$
which establishes Theorem~\ref{thm:logT} since $r_i^{-1} \rightarrow
\infty$ as $i \rightarrow \infty$.

\section{Tractability via more intuitive oracle access}
\label{sec:simpler-alg}

\newcommand{\LimSet}{\ensuremath{\text{\sc lim}}}
\newcommand{\Decomposable}{Cantor-Bendixson}

In Theorem~\ref{thm:main-alg}, the algorithm accesses the metric space via two oracles: a very intuitive \emph{covering oracle}, and a less intuitive \emph{ordering oracle}. In this section we show that for a wide family of metric spaces --- including, for example, compact metric spaces with a finite number of limit points --- the ordering oracle is not needed: we provide an algorithm which accesses the metric space via a finite set of covering oracles. We will consider metric spaces of finite \emph{Cantor-Bendixson rank}, a classic notion from point topology.

\begin{definition}\label{def:CB-rank}
Fix a metric space $(X,d)$. If for some $x\in X$ there exists a sequence of points in $X\setminus \{x\} $ which converges to $x$, then $x$ is called a \emph{limit point}. For $S\subset X$ let $\LimSet(S)$ denote the \emph{limit set}: the set of all limit points of $S$. Let
	$\LimSet(S,0) = S$,
and
	$\LimSet(S,i) = \LimSet(\LimSet(\cdots \LimSet(S)))$,
where $\LimSet(\cdot)$ is applied $i$ times.
The \emph{Cantor-Bendixson rank} of $(X,d)$ is defined as
    $\sup \{n: \LimSet(X,n) \neq \emptyset\}$.
\end{definition}

Let us say that a \emph{\Decomposable\ metric space} is one with a finite Cantor-Bendixson rank.\ASnote{Defining a ``Cantor-Bendixson metric space" is very clever, isn't it? :)} In order to apply Theorem~\ref{thm:main-alg}, we show that any such metric space is well-orderable.

\begin{lemma}
Any \Decomposable\ metric space is well-orderable.
\end{lemma}

\begin{proof}
Any finite metric space is trivially well-orderable. To prove the lemma, it suffices to show the following: any metric space $(X,d)$ is well-orderable  if so is $(\LimSet(X),d)$.

Let $X_1 = X\setminus  \LimSet(X)$ and $X_2 = \LimSet(X)$. Suppose $(X_2,d)$ admits a topological well-ordering $\prec_2$. Define a binary relation $\prec$ on $X$ as follows. Fix an arbitrary well-ordering $\prec_1$ on $X_1$. For any $x,y\in X$ posit $x\prec y$ if either
(i) $x,y\in X_1$ and $x\prec_1 y$, or
(ii) $x,y\in X_2$ and $x\prec_2 y$, or
(iii) $x\in X_1$ and $y\in X_2$.
It is easy to see that $(X,\prec)$ is a well-ordering.

It remains to prove that an arbitrary initial segment
	$Y = \{ x\in X: x\prec y\}$
is open in $(X,d)$. We need to show that for each $x\in Y$ there is a ball $B(x,\eps)$, $\eps>0$ which is contained in $Y$. This is true if $x\in X_1$ since by definition each such $x$ is an isolated point in $X$. If $x\in X_2$ then
	$Y = X_1 \cup Y_2$
where $Y_2 = \{ x\in X_2: x \prec_2 y\}$
is the initial segment of $X_2$. Since $Y_2$ is open in $(X_2,d)$, there exists $\eps>0$ such that $B_{X_2}(x,\eps) \subset Y_2$. It follows that
	$B_X(x,\eps) \subset B_{X_2}(x,\eps) \cup X_1 \subset Y$.
\end{proof}

The structure of a \Decomposable\ metric space is revealed by a partition of $X$ into subsets
	$X_i = \LimSet(X,i) \setminus \LimSet(X,i+1)$, $0\leq i\leq n $.
For a point $x\in X_i$, we define the \emph{rank} to be $i$. The algorithm requires a covering oracle for each $X_i$.

\begin{theorem}\label{thm:lim-decomposable}
Consider the Lipschitz MAB/experts problem on a compact metric space $(X,d)$ such that
	$\LimSet_N(X)=\emptyset$
for some $N$.  Let $\mathcal{O}_i$ be the covering oracle for
	$X_i = \LimSet(X,i)\setminus \LimSet(X,i+1)$.
Assume that access to the metric space is provided only via the collection of oracles
	$\{\mathcal{O}_i \}_{i=0}^N$.
Then:
\begin{OneLiners}
\item[(a)] the Lipschitz MAB problem on $(X,d)$ is $f$-tractable for every $f\in\omega(\log t)$;
\item[(b)] the Lipschitz experts problem on $(X,d)$ is 1-tractable, even with a double feedback.
\end{OneLiners}
\end{theorem}

\OMIT{ 
\begin{theorem}\label{thm:lim-decomposable}
Consider a compact metric space $(X,d)$ such that
	$\LimSet_n(X)=\emptyset$
for some $n$.  Let $\mathcal{O}_i$ be the covering oracle for
	$X_i = \LimSet(X,i)\setminus \LimSet(X,i+1)$.
Then for each $f\in \omega(\log t)$ there exists an $f(t)$-tractable bandit algorithm which accesses $(X,d)$ via the collection of oracles
	$\{\mathcal{O}_i \}_{i=0}^n$.
\end{theorem}
} 

In the rest of this section, consider the setting in Theorem~\ref{thm:lim-decomposable}. We describe the \emph{exploration subroutine} $\explSub'()$, which is similar to $\explSub()$ in Section~\ref{sec:tractability} but does not use the ordering oracle. Then we prove a version of Lemma~\ref{lm:tractability} for $\explSub'()$. Once we have this lemma, the proof of Theorem~\ref{thm:lim-decomposable} is identical to that of Theorem~\ref{thm:main-alg} (and is omitted).

\begin{algorithm}\label{alg:LIM}
Subroutine $\explSub'(k,n,r)$: inputs $k, n \in \N$ and $r\in (0,1)$, outputs a point in $X$.

Call each covering oracle $\mathcal{O}_i(k)$ and receive a $\delta_i$-covering $S_i$ of $X$ consisting of at most $k$ points. Let
	$S = \cup_{l=1}^n S_l$.
Play each strategy $x\in S$ exactly $n$ times; let $\muAv(x)$ be the corresponding sample average. For $x,y\in S$, let us say that \emph{$x$ dominates $y$} if
	$\muAv(x) - \muAv(y)> 2\,r$.
Call $x\in S$ a \emph{winner} if $x$ has a largest rank among the strategies that are not dominated by any other strategy. Output an arbitrary winner if a winner exists, else output an arbitrary point in $S$.
\end{algorithm}

Clearly, $\explSub(k,n,r)$ takes at most $knN$ rounds to complete. We show that for sufficiently large $k,n$ and sufficiently small $r$ it returns an optimal strategy with high probability.

\begin{lemma}\label{lm:LIM-tractability}
Fix a problem instance. Consider increasing functions $k,n,T: \N\to \N$ such that
    $r(t) := 4\sqrt{ (\log T(t))\, / n(t)}  \to 0$.
Then for any sufficiently large $t$, with probability at least $1-T^{-2}(t)$, the subroutine
    $\explSub'(k(t),\,n(t),\, r(t))$ returns an optimal strategy.
\end{lemma}

\begin{proof}
Use the notation from Algorithm~\ref{alg:LIM}. Fix $t$ and consider a run of
    $\explSub'(k(t),\,n(t),\, r(t))$.
Call this run \emph{clean} if for each $x\in S$ we have $|\muAv(x)-\mu(x)| \leq r(t)$. By Chernoff Bounds, this happens with probability at least $1-T^{-2}(t)$. In the rest of the proof, let us assume that the run is clean.

Let us introduce some notation. Let $\mu$ be the payoff function and let $\mu^* = \sup(\mu,X)$. Call $x\in X$ \emph{optimal} if $\mu(x) = \mu^*$. (There exists an optimal strategy since $(X,d)$ is compact.) Let $i^*$ be the largest rank of any optimal strategy. Let $X^*$ be the set of all optimal strategies of rank $i^*$. Let $Y = \LimSet(X,i^*)$. Since each point $x\in X_{i^*}$ is an isolated point in $Y$, there exists some $r(x)>0$ such that $x$ is the only point of $B(x,r(x))$ that lies in $Y$.

We claim that $\sup(\mu, Y\setminus X^*)<\mu^*$. Indeed, consider
	$C= \cup_{x\in X^*} B(x,r(x)) $.
This is an open set. Since $Y$ is closed, $Y\setminus C$ is closed, too, hence compact. Therefore there exists $y\in Y\setminus C$ such that
	$\mu(y) = \sup(\mu, Y\setminus C)$.
Since $X^* \subset C$, $\mu(y)$ is not optimal, i.e. $\mu(y)<\mu^*$. Finally, by definition of $r(x)$ we have $Y\setminus C = Y\setminus X^*$. Claim proved.

Pick any $x^*\in X^*$. Let $\mu_0 = \sup(\mu, Y\setminus X^*)$.
Assume that $t$ is large enough so that $r(t) < (\mu^*-\mu_0)/4$ and $\delta_{i^*}<r(x^*)$. Note that the $\delta_{i^*}$-covering $S_{i^*}$ contains $x^*$.

Finally, we claim that in a clean phase, $x^*$ is a winner, and all winners lie in $X^*$. Indeed, note that $x^*$ dominates any non-optimal strategy $y\in S$ of larger or equal rank, i.e. any $y\in S\cap (Y\setminus X^*)$. This is because
$ \muAv(x^*) - \muAv(y) \geq \mu^* - \mu_0 - 2r > 2.
$
The claim follows since any optimal strategy cannot be dominated by any other strategy.
\end{proof}

\section{Boundary of tractability: Theorem~\ref{thm:boundary-of-tractability}}
\label{sec:boundary-body}

In this section we prove Theorem~\ref{thm:boundary-of-tractability}. In Appendix~\ref{sec:reduction} we reduce the theorem to that on complete metric spaces. We will use a basic fact that a complete metric space is compact if and only if for any $r>0$, it can be covered by a finite number of balls of radius $r$.

\xhdr{Algorithmic result.}
We consider a compact metric space $(X,d)$ and use an extension of the \emph{naive algorithm} from~\cite{Bobby-nips04, LipschitzMAB-stoc08}. In each phase $i$ (which lasts for $t_i$ round) we fix a covering of $X$ with $N_i < \infty$ balls of radius $2^{-i}$ (such covering exists by compactness), and run a fresh instance of the $N_i$-armed bandit algorithm \UCB\ from~\cite{bandits-ucb1} on the centers of these balls.  (This algorithm is for the ``basic'' MAB problem,
in the sense that it does not look at the distances in the metric
space.) The phase durations $t_i$ need to be tuned to the $N_i$'s. In the setting considered in~\cite{Bobby-nips04, LipschitzMAB-stoc08} (bounded covering dimension) it suffices to tune each $t_i$ to the corresponding $t_i$ in a fairly natural way. The difficulty in the present setting is that there are no guarantees on how fast the $N_i$'s grow. To take this into account, we fine-tune each $t_i$ to (essentially) all covering numbers $N_1 \LDOTS N_{i+1}$.

Let $R_k(t)$ be the expected regret accumulated by the algorithm in the first $t$ rounds of phase $k$. Using the off-the-shelf regret guarantees for \UCB, it is easy to see~\cite{Bobby-nips04, LipschitzMAB-stoc08} that
\begin{align}\label{eq:app-boundary-UB-Rk}
 R_k(t) \leq O(\sqrt{N_k\, t \log t}) + \eps_k\, t
\leq \eps_k\, \max(t^*_k,\, t),
    \text{~~where~~} t^*_k = 2\,\tfrac{N_k}{\eps_k^2} \log\tfrac{N_k}{\eps_k^2}.
\end{align}

Let us specify phase durations $t_i$. They are defined very differently from the ones in~\cite{Bobby-nips04, LipschitzMAB-stoc08}. In particular, in~\cite{Bobby-nips04, LipschitzMAB-stoc08} each $t_i$ is fine-tuned to the corresponding covering number $N_i$ by setting
    $t_i = t^*_i$,
and the analysis works out for metric spaces of bounded covering dimension. In our setting, we fine-tune each $t_i$ to (essentially) all covering numbers $N_1 \LDOTS N_{i+1}$. Specifically, we define the $t_i$'s inductively as follows:
$$ t_i = \min(t^*_i,\, t^*_{i+1},\, 2\,
    \textstyle{\sum_{j=1}^{i-1}} t_j).
$$
This completes the description of the algorithm, call it \A.

\begin{lemma}
Consider the Lipschitz MAB problem on a compact and complete metric space $(X,d)$. Then $R_\A(t) \leq 5\, \eps(t)\, t$, where
    $\eps(t) = \min\{ 2^{-k}:\, t\leq s_k\}$
and
    $s_k = \sum_{i=1}^k\, t_i$.
In particular, $R_\A(t) = o(t)$.
\end{lemma}

\begin{proof}
First we claim that
    $R_\A(s_k) \leq 2\, \eps_k\, s_k$
for each $k$. Use induction on $k$. For the induction base, note that
    $R_\A(s_1) = R_1 (t_1) \leq \eps_1 t_1$
by~\refeq{eq:app-boundary-UB-Rk}. Assume the claim holds for some $k-1$. Then
\begin{align*}
R_\A(s_k)
    &= R_\A(s_{k-1}) + R_k(t_k) \\
    &\leq 2\,\eps_{k-1}\, s_{k-1} + \eps_k\, t_k \\
    &\leq 2\, \eps_k (s_{k-1} + t_k)
    = 2\,\eps_k s_k,
\end{align*}
claim proved. Note that we have used~\refeq{eq:app-boundary-UB-Rk} and the facts that
    $t_k \geq t^*_k$ and $t_k \geq 2\, s_{k-1}$.

For the general case, let $T = s_{k-1} + t$, where $t\in (0, t_k)$. Then by~\refeq{eq:app-boundary-UB-Rk} we have that
\begin{align*}
R_k(t)
    &\leq \eps_k\, \max(t^*_k,\, t) \\
    &\leq \eps_k\, \max(t_{k-1},\, t)
    \leq \eps_k\, T,\\
R_\A(T) &= R_\A(s_{k-1}) + R_k(T)  \\
    &\leq 2\, \eps_{k-1}\, s_{k-1} + \eps_k\,T
    \leq 5\,\eps_k\, T. \qedhere
\end{align*}
\end{proof}

\xhdr{Lower bound: proof sketch.} For the lower bound, we consider a metric space $(X,d)$ with an infinitely many disjoint balls $B(x_i,r_*)$ for some $r_*>0$. For each ball $i$ we define the \emph{wedge function} supported on this ball:
$$ G_{(i,r)}(x) = \begin{cases}
    \min\{ r_* - d(x,x_i),\; r_* - r \} & \mbox{if $x \in B(x_i,r_*)$} \\
         0 & \mbox{otherwise}.
\end{cases} $$
The balls are partitioned into two infinite sets: the
\emph{ordinary} and \emph{special} balls.
The random payoff function is then defined by taking a
constant function, adding the wedge function on each
special ball, and randomly adding
or subtracting the wedge function on each
ordinary ball.
Thus, the expected payoff is constant throughout the
metric space except that it assumes higher values
on the special balls.  However,
the algorithm has no chance of ever finding these balls,
because at time $t$ they are statistically indistinguishable
from the $2^{-t}$ fraction of ordinary balls that randomly
happen to never subtract their wedge function during the
first $t$ steps of play.

\xhdr{Lower bound: full proof.}
Suppose $(X,d)$ is not compact. Fix $r>0$ such that $X$ cannot be covered by a finite number of balls of radius $r$. There exists a countably infinite subset $S\subset X$ such that the balls $B(x,r)$, $x\in S$ are mutually disjoint. (Such subset can be constructed inductively.)
Number the elements of $S$ as $s_1,s_2,\ldots,$ and denote
the ball $B(s_i,r)$ by $B(i)$.

Suppose there exists a Lipschitz experts algorithm
$\A$ that is $g(t)$-tractable for some $g\in o(t)$.
Pick an increasing sequence
    $t_1, t_2, \ldots \in \N$
such that $t_{k+1} > 2 t_k \geq 10$ and
    $g(t_k)< r_k\, t_k /k$
for each $k$, where
    $r_k = r/2^{k+1}$.
Let $m_0=0$ and
    $m_k = \sum_{i=1}^{k} 4^{t_i}$
for $k>0$, and let $I_k = \{m_{k}+1,\ldots,m_{k+1}\}.$
The intervals $I_k$ form a partition of $\N$ into
sets of sizes $4^{t_1},4^{t_2},\ldots$.
For every $i \in \N$, let $k$ be the unique value
such that $i \in I_k$ and define the following
Lipschitz function supported in $B(s_i,r)$:
    $$ G_i(x) = \begin{cases}
         \min\{ r - d(x,s_i), r - r_k \} & \mbox{if $x \in B(i)$} \\
         0 & \mbox{otherwise}.
       \end{cases} $$
If $J \subseteq \N$ is any set of natural numbers,
we can define a distribution $\prob_J$ on payoff
functions by sampling independent, uniformly-random
signs $\sigma_i \in \{ \pm 1 \}$ for every $i \in \N$
and defining the payoff function to be
    $$ \payoff = \tfrac12 + \textstyle{\sum_{i \in J}}\, G_i
       + \textstyle{\sum_{i \not\in J}}\, \sigma_i G_i.
    $$
Note that the distribution $\prob_J$ has expected
payoff function $\mu = \tfrac12 + \sum_{i \in J} G_i.$
Let us define a distribution $\mathcal{P}$ over problem
instances $\prob_J$ by letting $J$ be a random
subset of $\N$ obtained by sampling exactly one element $j_k$
of each set $I_k$ uniformly at random, independently
for each $k$.

Intuitively, consider an algorithm that
is trying to discover the value of $j_k$.  Every time a payoff function
$\pi_t$ is revealed, we get to see a random $\{ \pm 1 \}$
sample at every element of $I_k$ and we can eliminate
the possibility that $j_k$ is one of the elements that
sampled $-1$.  This filters out about half the elements
of $I_k$ in every time step, but $|I_k| = 4^{t_k}$ so
on average it takes $2 t_k$ steps before we can discover the identity
of $j_k$.  Until that time, whenever we play a strategy
in $\cup_{i \in I_k} B(i)$, there is a constant probability
that our regret is at least $r_k$.  Thus our regret
is bounded below by $r_k t_k \geq k g(t_k).$  This
rules out the possibility of a $g(t)$-tractable
algorithm.  The following lemma makes this argument
precise.

\begin{lemma}\label{lm:app-boundary}
$\Pr_{\prob \in\mathcal{P}} [ R_{(\A, \, \prob)} (t) = O_\mu( g(t)) ] = 0$.
\end{lemma}
\begin{proof}
Let $j_1,j_2,\ldots$ be the elements of the
random set $J$, numbered
so that $j_k \in I_k$ for all $k$.
For any $i,t \in \N$, let $\sigma(i,t)$ denote
the value of $\sigma_i$ sampled at time $t$
when sampling the sequence of i.i.d.~payoff
functions $\payoff_t$ from distribution $\prob_J$.
We know that $\sigma(j_k,t)=1$ for all $t$.
In fact if $S(k,t)$ denotes the set of all
$i \in I_k$ such that
    $\sigma(i,1) = \sigma(i,2) = \cdots = \sigma(i,t) = 1$
then conditional on the value of the set $S(k,t)$, the
value of $j_k$ is distributed uniformly at random
in $S(k,t)$.  As long as this set $S(k,t)$
has at least $n$ elements, the probability
that the algorithm picks a strategy $x_t$ belonging
to $B(j_k)$ at time $t$ is bounded above by $\tfrac{1}{n}$,
even if we condition on the event that
    $x_t \in \cup_{i \in I_k} B(i).$
For any given $i \in I_k \setminus \{j_k\}$, we have
    $\prob_J(i \in S(k,t)) = 2^{-t}$
and these events are independent for different
values of $i$.  Setting $n=2^{t_k}$, so that
$|I_k|=n^2$, we have
    \begin{align}
      \nonumber
    \prob_J \left[\, |S(k,t)| \leq n \,\right] &\leq
    \textstyle{\sum_{R \subset I_k,\; |R|=n}}\; \prob_J[\, S(k,t) \subseteq R \,] \\
      \nonumber
    &= \binom{n^2}{n}
    \left( 1-2^{-t} \right)^{n^2-n}
      \nonumber
    < \left( n^2 \cdot \left( 1 - 2^{-t} \right)^{n-1} \right)^{n} \\
      \label{eq:yucky-1}
    &< \exp \left( n (2 \ln(n) - (n-1)/2^t) \right).
    \end{align}
As long as $t \leq t_{k-1}$, the relation $t_k > 2t$
implies $(n-1)/2^t > \sqrt{n}$ so the expression
\eqref{eq:yucky-1} is bounded above by
$\exp \left( -n \sqrt{n} + 2 n \ln(n) \right)$,
which equals
$\exp \left( -8^{t_k} + 2 \ln(4) t_k 4^{t_k} \right)$
and is in turn bounded above by $\exp \left( -8^{t_k}/2 \right).$

Let $B(j_{>k})$ denote the union
    $B(j_{k+1}) \cup B(j_{k+2}) \cup \ldots,$
and let $N(t,k)$ denote the random variable
that counts the number of times \A\ selects
a strategy in $B(j_{>k})$ during rounds $1,\ldots,t$.
We have already demonstrated that for all $t \leq t_k$,
    \begin{equation} \label{eq:yucky-2}
       \Pr_{\prob_J \in \mathcal{P}}(x_t \in B(j_{>k}))
       \leq 2^{-t_{k+1}} + \sum_{\ell > k}
       \exp \left( -8^{t_\ell}/2 \right) < 2^{1-t_{k+1}},
    \end{equation}
where the term $2^{-t_{k+1}}$ accounts for the
event that $S(\ell,t)$ has at least $2^{t_{k+1}}$ elements,
where $\ell$ in the index of the set $I_{\ell}$ containing
the number $i$ such that $x_t \in B(i)$, if such an $i$
exists.
Equation \eqref{eq:yucky-2} implies the bound
    $\mathbb{E}_{\prob_J \in \mathcal{P}}[N(t_k,k)] < t_k \cdot 2^{1-t_{k+1}}.$
By Markov's inequality, the probability that $N(t_k,k) > t_k/2$
is less than $2^{2-t_{k+1}}$.  By Borel-Cantelli, almost
surely the number of $k$ such that $N(t_k,k) \leq t_k/2$
is finite.
The algorithm's expected regret at time $t$ is
bounded below by $r_k(t_k - N(t_k,k))$, so with
probability $1$, for all but finitely many $k$ we have
$R_{(\A, \, \prob_J)}(t_k) \geq r_k t_k / 2 \geq (k/2) g(t_k).$
This establishes that \A\ is not $g(t)$-tractable.
\end{proof}

\section{Lipschitz experts in a (very) high dimension}
\label{sec:FFproblem}

In this section we discuss the \FFproblem\ in (very) high dimensional metric spaces. We posit a new notion of dimensionality which is well-suited to describe regret in such problems, provide examples of metric spaces for which this notion is relevant (Section~\ref{sec:LCD-examples}), and analyze the performance of a simple algorithm in terms of this notion (Section~\ref{sec:LCD-naive}). Moreover, we consider the same algorithm under a somewhat restricted version of the problem, and obtain much better regret guarantees via a more involved analysis (Section~\ref{sec:ULproblem}).

Fix a metric space $(X,d)$. For a subset $Y\subset X$ and $\delta>0$, a \emph{$\delta$-covering} of $Y$ is a collection of sets of diameter at most $\delta$ whose union contains $Y$. A subset $S\subset X$ is a \emph{$\delta$-hitting set} for $Y$ if
    $Y \subset \cup_{x\in S}\, B(x,\,\delta)$.
(So if $S$ is a hitting set for some $\delta$-covering of $Y$ then it is a $\delta$-hitting set for $Y$.)

Let $N_\delta(Y)$ be the minimal size (cardinality) of a $\delta$-covering of $Y$, i.e. the smallest number of sets of diameter at most $\delta$ sufficient to cover $Y$. The standard definition of the \emph{covering dimension} is
\begin{align}\label{eq:CovDim'}
    \Cov(Y) = \limsup_{\delta>0}\, \frac{\log N_\delta(Y)}{\log (1/\delta)}.
\end{align}
Covering dimension and its refinements have been essential in the study of the Lipschitz MAB problem~\cite{LipschitzMAB-stoc08}. However, for the \FFproblem\ the metrics with bounded covering dimension are too ``easy''. We need to consider a much broader class of metrics that satisfy a non-trivial bound on what we call the \emph{log-covering-dimension}:
\begin{align}\label{eq:CovDim}
    \LCD(Y) = \limsup_{\delta>0}\, \frac{\log\log N_\delta(Y)}{\log (1/\delta)}.
\end{align}

\subsection{Log-covering dimension: some examples}
\label{sec:LCD-examples}

To give an example of a metric space with a non-trivial log-covering dimension, let us consider a \emph{uniform tree} -- a rooted tree in which all nodes at the same level have the same number of children. An \emph{\eps-uniform tree metric} is a metric on the leaves of an infinitely deep uniform tree, in which the distance between two leaves is $\eps^{-i}$, where $i$ is the level of their least common ancestor. It is easy to see that an \eps-uniform tree metric such that the branching factor at each level $i$ is
    $ \exp(\eps^{-ib} (2^b - 1))$
has log-covering dimension $b$.

For another example, fix a metric space $(X,d)$ of finite diameter, and let $\mathcal{P}_X$ denote the set of all probability measures over $X$. Consider $\mathcal{P}_X$ as a metric space $(\mathcal{P}_X,W_1)$ under the Wasserstein $W_1$ metric, a.k.a. the Earthmover distance.\footnote{For a metric space $(X,d)$ of finite diameter, and two probability measures $\mu$, $\nu$ on $X$ the \emph{Wasserstein $W_1$ distance},  a.k.a. the \emph{Earthmover distance}, is defined as $W_1(\mu,\nu) = \inf \expect[|X-Y|]$, where the infimum is taken over all simultaneous distributions of the random variables X and Y with marginals $\mu$ and $\nu$ respectively. The Wasserstein distance defines a metric space on $(\mathcal{P}_X, W_1)$. It is one of the standard ways to define a distance on probability measures. In particular, it is widely used in Computer Science literature to compare discrete distributions, e.g. in the context of image retrieval~\cite{earthmover-00}.} We claim that the log-covering dimension of this metric space is equal to the covering dimension of $(X,d)$.

\begin{theorem}\label{thm:LCD-earthmover}
Let $(X,d)$ be a metric space of finite diameter whose covering dimension is $\kappa<\infty$. Let $(\mathcal{P}_X,W_1)$ be the space of all probability measures over $(X,d)$ under the Wasserstein $W_1$ metric. Then
	$\LCD(\mathcal{P}_X,W_1)=\kappa$.
\end{theorem}

In the remainder of this subsection we prove Theorem~\ref{thm:LCD-earthmover}.

\begin{proof}[Proof (Theorem~\ref{thm:LCD-earthmover}: upper bound)]
Let us cover $(\mathcal{P}_X,W_1)$ with balls of radius $\tfrac{2}{k}$ for some $k\in\N$. Let S be a $\tfrac{1}{k}$-net in $(X,d)$; note that $|S|=O(k^\kappa)$ for a sufficiently large $k$. Let $P$ be the set of all probability distributions $p$ on $(X,d)$ such that $\texttt{support}(p) \subset S$ and for every point $x \in S$, $p(x)$ is a rational number with denominator $k^{d+1}$. The cardinality of $P$ is bounded above by $(k^{\kappa+1})^{k^\kappa}$. It remains to show that balls of radius $\tfrac{2}{k}$ centered at the points of $P$ cover the entire space $(\mathcal{P}_X,W_1)$. This is true because:
\begin{itemize}
\item every distribution $q$ is $\tfrac{1}{k}$-close to a distribution $p$ with support contained in $S$ (let $p$ be the distribution defined by randomly sampling a point of $(X,d)$ from $q$ and then outputting the closest point of $S$);
\item every distribution with support contained in $S$ is $\tfrac{1}{k}$-close to a distribution in $P$ (round all probabilities down to the nearest multiple of $k^{-(\kappa+1)}$; this requires moving only $\tfrac{1}{k}$ units of stuff). \qedhere
\end{itemize}
\end{proof}

To prove the lower bound, we make a connection to the Hamming metric.

\begin{lemma}  \label{lem:embedding}
Let $(X,d)$ be any metric space, and let $H$ denote the Hamming metric on the Boolean cube $\{0,1\}^n$.
If $S \subseteq X$ is a subset of even cardinality $2n$,
and $\eps$ is a lower bound on the distance between any two points
of $S$, then there is a mapping $f \,:\, \{0,1\}^n \rightarrow
\mathcal{P}_X$ such that for all $a,b \in \{0,1\}^n,$
\begin{align} \label{eq:earthmover-lb}
W_1(f(a),f(b)) \geq \tfrac{\eps}{n}\; H(a,b).
\end{align}
\end{lemma}
\begin{proof}
Group the points of $S$ arbitrarily into pairs $S_i = \{x_i,y_i\}$,
where $i=1,\ldots,n.$  For $a \in \{0,1\}^n$ and $1 \leq i \leq n$,
define $t_i(a) = x_i$ if $a_i=0$, and $t_i(a) = y_i$ otherwise.
Let $f(a)$ be the uniform distribution on the set
$\{t_1(a),\ldots,t_n(a)\}.$
To prove \eqref{eq:earthmover-lb}, note that if
$i$ is any index such that $a_i \neq b_i$ then
$f(a)$ assigns probability $\tfrac{1}{n}$ to $t_i(a)$ while
$f(b)$ assigns zero probability to the entire ball
of radius $\eps$ centered at $t_i(a).$  Consequently,
the $\tfrac{1}{n}$ units of probability at $t_i(a)$ have to
move a distance of at least $\eps$ when shifting
from distribution $f(a)$ to $f(b)$.  Summing over
all indices $i$ such that $a_i \neq b_i$, we obtain~\refeq{eq:earthmover-lb}.
\end{proof}

The following lemma, asserting the existence of
asymptotically good binary error-correcting codes,
is well known, e.g. see~\cite{Gilbert,Varshamov}.

\begin{lemma} \label{lem:asympt-good}
Suppose $\delta, \rho$ are constants satisfying
$0 < \delta < \frac12$ and
$0 \leq \rho < 1 + \delta \log_2(\delta) + (1-\delta) \log_2(1-\delta).$
For every sufficiently large $n$, the Hamming cube
$\{0,1\}^n$ contains more than $2^{\rho n}$ points,
no two of which are nearer than distance $\delta n$
in the Hamming metric.
\end{lemma}

Combining these two lemmas, we obtain an easy proof for the lower bound in Theorem~\ref{thm:LCD-earthmover}.

\begin{proof}[Proof (Theorem~\ref{thm:LCD-earthmover}: lower bound)]
Consider any $\gamma < \kappa$.  The hypothesis on the covering
dimension of $(X,d)$ implies that for all sufficiently small
$\eps$, there exists a set $S$ of cardinality $2n$ --- for some
$n > \eps^{-\gamma}$ --- such that the minimum distance between
two points of $S$ is at least $5 \eps.$  Now let $\mathcal{C}$
be a subset of $\{0,1\}^n$ having at least $2^{n/5}$ elements,
such that the Hamming distance between any two points of $\mathcal{C}$
is at least $n/5.$  Lemma~\ref{lem:asympt-good} implies that
such a set $\mathcal{C}$ exists, and we can then apply
Lemma~\ref{lem:embedding} to embed $\mathcal{C}$ in
$\mathcal{P}_X$, obtaining a subset of $\mathcal{P}_X$ whose
cardinality is at least $2^{\eps^{-\gamma} / 5}$, with
distance at least $\eps$ between every pair of points in the set.
Thus, any $\eps$-covering of $\mathcal{P}_X$ must contain at
least $2^{\eps^{-\gamma}/5}$ sets, implying that
$\LCD(\mathcal{P}_X,W_1) \geq \gamma.$  As $\gamma$ was an
arbitrary number less than $\kappa$, the proposition is proved.
\end{proof}

\newcommand{\NaiveExp}{\ensuremath{\text{{\sc NaiveExperts}}}}

\subsection{Using the new definition}
\label{sec:LCD-naive}

To see how the log-covering dimension is relevant to the Lipschitz experts problem, consider the following simple algorithm, called \NaiveExp(b).\footnote{This algorithm is a version of the ``naive algorithm''~\cite{Bobby-nips04,LipschitzMAB-stoc08} for the Lipschitz MAB problem. A similar algorithm has been used by~\cite{Anupam-experts07} to obtain regret $R(T) = O(\sqrt{T})$ for metric spaces of finite covering dimension.} The algorithm is parameterized by $b>0$. It runs in phases. Each phase $i$ lasts for $T = 2^i$ rounds, and outputs its \emph{best guess} $x^*_i\in X$, which is played throughout phase $i+1$. During phase $i$, the algorithm picks a $\delta$-hitting set for $X$ of size at most $N_\delta(X)$, for
    $\delta = T^{-1/(b+2)}$.
By the end of the phase, $x^*_i$ as defined as the point in $S$ with the highest sample average (breaking ties arbitrarily). This completes the description of the algorithm.

It is easy to see that the regret of \NaiveExp\ is naturally described in terms of the log-covering dimension. The proof is based the argument from~\cite{Bobby-nips04}. We restate it here for the sake of completeness, and to explain how the new dimensionality notion is used.

\OMIT{ 
It is easy to see that for any $b>\LCD(X)$ algorithm \NaiveExp(b)\ achieves regret
    $R(t) = O(t^{1-1/(b+2)})$
for the \FFproblem.
} 

\begin{theorem}\label{thm:ffproblem-naive}
Consider the \FFproblem\ on a metric space $(X,d)$. For each $b>\LCD(X)$, algorithm $\NaiveExp(b)$ achieves regret
    $R(t) = O(t^{1-1/(b+2)})$.
\end{theorem}

\begin{proof}
Let $N_\delta = N_\delta(X)$, and let $\mu$ be the expected payoff function. Consider a given phase $i$ of the algorithm. Let $T=2^i$ be the phase duration. Let $\delta = T^{-1/(b+2)}$, and let $S\subset X$ the $\delta$-hitting set chosen in this phase. Note that for any sufficiently large $T$ it is the case that
    $N_\delta < 2^{\delta^{-b}}$.
For each $x\in S$, let $\mu_T(x)$ be the sample average of the feedback from $x$ by the end of the phase. Then by Chernoff bounds,
\begin{align}\label{eq:ffproblem-naive}
 \Pr[ | \mu_T(x) - \mu(x)| <  r_T ] > 1-  (T N_\delta)^{-3},
    \quad\text{where}\quad
        r_T = \sqrt{ 8\,\log (T\, N_\delta)\, / T} < 2\delta.
\end{align}
Note that $\delta$ is chosen specifically to ensure that $r_T \leq O(\delta)$.

We can neglect the regret incurred when the event in~\refeq{eq:ffproblem-naive} does not hold for some $x\in S$. From now on, let us assume that the event in~\refeq{eq:ffproblem-naive} holds for all $x\in S$. Let $x^*$ be an optimal strategy, and $x^*_i = \argmax_{x\in S} \mu_T(x) $ be the ``best guess". Let $x\in S$ be a point that covers $x^*$. Then
$$ \mu(x^*_i)
    \geq \mu_T(x^*_i) - 2\delta
    \geq \mu_T(x) - 2\delta
    \geq \mu(x) - 4 \delta
    \geq \mu(x^*) - 5 \delta.
$$
Thus the total regret $R_{i+1}$ accumulated in phase $i+1$ is
$$ R_{i+1}\leq 2^{i+1}\, (\mu(x^*) - \mu(x^*_i)) \leq O(\delta T)
    = O(T^{1-1/(2+b)}).
$$
Thus the total regret summed over phases is as claimed.
\end{proof}

\subsection{The \ULproblem}
\label{sec:ULproblem}

We now turn our attention to the \emph{\ULproblem}, a restricted version of the \FFproblem\ in which a problem instance $(X,d,\prob)$ satisfies a further property that each function $f\in \mathtt{support}(\mathbb{P})$ is itself a Lipschitz function on $(X,d)$. We show that for this version, \NaiveExp\ obtains a significantly better regret guarantee, via a more involved analysis. As we will see in the next section, for a wide class of metric spaces including \eps-uniform tree metrics there is a matching upper bound.

\OMIT{ 
The lower bound in Theorem~\ref{thm:experts-MaxMinLCD} holds for the \emph{\ULproblem},
Let us show the matching upper bound for the metric spaces such that
    $\MaxMinLCD(X) = \LCD(X)$.
In fact, this bound is achieved by algorithm $\NaiveExp$ via a more involved analysis.
} 

\begin{theorem}\label{thm:ulproblem-naive}
Consider the \ULproblem\ with full feedback. Fix a metric space $(X,d)$. For each $b>\LCD(X)$ such that $b\geq 2$, $\NaiveExp(b-2)$ achieves regret
    $R(t) = O(t^{1-1/b})$.
\end{theorem}

\begin{proof}
The preliminaries are similar to those in the proof of Theorem~\ref{thm:ffproblem-naive}. For simplicity, assume $b\geq 2$. Let $N_\delta = N_\delta(X)$, and let $\mu$ be the expected payoff function. Consider a given phase $i$ of the algorithm. Let $T=2^i$ be the phase duration. Let $\delta = T^{-1/b}$, and let $S$ be the $\delta$-hitting set chosen in this phase. (The specific choice of $\delta$ is the only difference between the algorithm here and the algorithm in Theorem~\ref{thm:ffproblem-naive}.) Note that $|S|\leq N_\delta$, and for any sufficiently large $T$ it is the case that
    $N_\delta < 2^{\delta^{-b}}$.

The rest of the analysis holds for any set $S$ such that $|S|\leq N_\delta$. (That is, it is not essential that $S$ is a $\delta$-hitting set for $X$.) For each $x\in S$, let $\nu(x)$ be the sample average of the feedback from $x$ by the end of the phase. Let $y^*_i = \argmax (\mu,S)$ be the optimal strategy in the chosen sample, and let $x^*_i = \argmax(\nu,S)$ be the algorithm's ``best guess". The crux is to show that
\begin{align}\label{eq:thm-ulproblem-naive-crux}
 \Pr[\, \mu(y^*_i) - \mu(x^*_i) \leq O(\delta \log T) \,]
     > 1-  T^{-3}.
\end{align}
Once~\refeq{eq:thm-ulproblem-naive-crux} is established, the remaining steps is exactly as the proof of Theorem~\ref{thm:ffproblem-naive}.

Proving~\refeq{eq:thm-ulproblem-naive-crux} requires a new technique. The obvious approach -- to use Chernoff Bounds for each $x\in S$ separately and then take a Union Bound -- does not work, essentially because one needs to take the Union Bound over too many points. Instead, we will use a more efficient version tail bound: for each $x,y\in X$, we will use Chernoff Bounds applied to the random variable $f(x)-f(y)$, where $f \sim \prob$ and  $(X,d,\prob)$ is the problem instance. For a more convenient notation, we define
$$\Delta(x,y)
    = \left[\, \mu(x)-\mu(y) \,\right]
    + \left[\, \nu(y)-\nu(x) \,\right],$$
Then for any $N\in \N$ we have
\begin{align}\label{eq:thm-ulproblem-naive-Chernoff}
 \Pr\left[\,  |\Delta(x,y)| \leq
        d(x,y)\,\sqrt{8\,\log (T\, N)/T}
 \right]
 > 1-  (T N)^{-3}.
\end{align}
The point is that the ``slack" in the Chernoff Bound is scaled by the factor of $d(x,y)$. This is because each $f\in \texttt{support}(\prob)$ is a Lipschitz function on $(X,d)$,

In order to take advantage of~\refeq{eq:thm-ulproblem-naive-Chernoff}, let us define the following structure that we call the \emph{covering tree} of the metric space $(X,d)$. This structure consists of a rooted tree $\mathcal{T}$ and non-empty subsets $X(u)\subset X$ for each internal node $u$. Let $V_\mathcal{T}$ be the set of all internal nodes. Let $\mathcal{T}_j$ be the set of all level-$j$ internal nodes (so that $\mathcal{T}_0$ is a singleton set containing the root). For each
    $u\in V_\mathcal{T}$,
let $\mathcal{C}(u)$ be the set of all children of $u$. For each node $u\in \mathcal{T}_j$ the structure satisfies the following two properties: (i) set $X(u)$ has diameter at most $2^{-j}$, (ii) the sets $X(v)$, $v\in \mathcal{C}(u)$ form a partition of $X(u)$. This completes the definition.

By definition of the covering number $N_\delta(\cdot)$ there exist a covering tree $\mathcal{T}$ in which each node $u\in \mathcal{T}_j$ has fan-out $N_{2^{-j}}( X(u))$. Fix one such covering tree. For each node $u\in V_\mathcal{T}$, define
\begin{align}
   \sigma(u) &= \argmax( \mu,\, \mathcal{X}(u) \cap S)
        \label{eq:thm-ulproblem-naive-sigma}\\
     \rho(u) &= \argmax( \nu,\, \mathcal{X}(u) \cap S), \nonumber
\end{align}
where the tie-breaking rule is the same as in the algorithm.

Let
    $n = \cel{\log\tfrac{1}{\delta}}$.
Let us say that phase $i$ is \emph{clean} if the following two properties hold:
\begin{OneLiners}
\item[(i)] for each node $u\in V_\mathcal{T}$  any two children $v,w\in \mathcal{C}(u)$ we have
    $|\, \Delta( \sigma(v),\, \sigma(w))\, | \leq 4\delta $.

\item[(ii)] for any $x,y\in S$ such that $d(x,y)\leq \delta$ we have
    $|\Delta(x,y)| \leq 4\delta $.
\end{OneLiners}

\begin{claim}
For any sufficiently large $i$,
phase $i$ is clean with probability at least $1-T^{-2}$.
\end{claim}
\begin{proof}
To prove (i), let $j$ be such that $u\in \mathcal{T}_j$. We consider each $j$ separately. Note that (i) is trivial for $j>n$. Now fix $j\leq n$ and apply the Chernoff-style bound~\refeq{eq:thm-ulproblem-naive-Chernoff} with
    $N = |\mathcal{T}_j|$ and $(x,y) = (\sigma(v), \sigma(w))$.
Since
	$|\mathcal{T}_l| \leq 2^{2^{lb}}\, |\mathcal{T}_{l-1}|$
for each sufficiently large $l$, it follows that
$\log |\mathcal{T}_j|
	\leq C + \textstyle{\sum_{l=1}^j}\; 2^{lb}
	\leq C+ \tfrac{4}{3}\, 2^{jb},
$
where $C$ is a constant that depends only on the metric space and $b$.
It is easy to check that for any sufficiently large phase $i$ (which, in turn, determines $T$, $\delta$ and $n$),  the ``slack" in ~\refeq{eq:thm-ulproblem-naive-Chernoff} is at most $4\delta$:
\begin{align*}
d(x,y)\,\sqrt{8\,\log (T\, N)/T}
	&\leq  3\, d(x,y)\,\sqrt{\log (N)/T}
	\leq 4\,2^{-j}\, \sqrt{2^{bj}/ 2^{bn}}
	= 4 \delta\, 2^{-(n-j)(b-2)/2}
	\leq 4\delta.
\end{align*}
Interestingly, the right-most inequality above is the only place in the proof where it is essential that $b\geq 2$.

To prove (ii), apply ~\refeq{eq:thm-ulproblem-naive-Chernoff} with $N = |S|$ similarly. Claim proved.
\end{proof}

From now on we will consider clean phase. (We can ignore regret incurred in the event that the phase is not clean.) We focus on the quantity
    $ \Delta^*(u) = \Delta( \sigma(u),\, \rho(u))$.
Note that by definition $ \Delta^*(u)\geq 0$. The central argument of this proof is the following upper bound on $\Delta^*(u)$.

\begin{claim}\label{cl:thm-ulproblem-naive}
In a clean phase,
    $\Delta^*(u) \leq O(\delta)(n-j)$
for each $j\leq n$ and each $u\in \mathcal{T}_j$.
\end{claim}
\begin{proof}
Use induction on $j$. The base case $j=n$ follows by part (ii) of the definition of the clean phase, since for $u\in \mathcal{T}_n$ both $\sigma(u)$ and $\rho(u)$ lie in $X(u)$, the set of diameter at most $\delta$. For the induction step, assume the claim holds for each $v\in \mathcal{T}_{j+1}$, and let us prove it for some fixed  $u\in \mathcal{T}_j$.

Pick children $u,v\in \mathcal{C}(u)$ such that
    $\sigma(u) \in X(v)$ and $\rho(u)\in X(w)$.
Since the tie-breaking rules in~\refeq{eq:thm-ulproblem-naive-sigma} is fixed for all nodes in the covering tree, it follows that
    $\sigma(u) = \sigma(v)$ and $\rho(u) = \rho(w)$.
Then
\begin{align*}
\Delta^*(w) + \Delta( \sigma(v),\, \sigma(w) )
    & = \Delta( \sigma(w),\, \rho(u) ) + \Delta( \sigma(u),\, \sigma(w) ) \\
    & =  \mu(\sigma(w)) - \mu(\rho(u)) + \nu(\rho(u)) - \rho(\sigma(w)) \;+\\
    & \quad \;
        \mu(\sigma(u)) - \mu(\sigma(w)) + \nu(\sigma(w)) - \nu(\sigma(u)) \\
    & = \Delta^*(u).
\end{align*}
Claim follows since
    $\Delta^*(w)\leq O(\delta)(n-j-1)$
by induction, and $\Delta( \sigma(v),\, \sigma(w) )\leq 4\delta$ by part (i) in the definition of the clean phase.
\end{proof}

To complete the proof of~\refeq{eq:thm-ulproblem-naive-crux}, let $u_0$ be the root of the covering tree. Then
    $y^*_i = \sigma(u_0)$ and $x^*_i = \rho(u_0)$.
Therefore by Claim~\ref{cl:thm-ulproblem-naive} (applied for $\mathcal{T}_0 = \{u_0\}$) we have
$$ O(\delta n) \geq \Delta^*(u_0) = \Delta^*(y^*_i, \, x^*_i)
    \geq \mu(y^*_i) - \mu(x^*_i). \qedhere
$$
\end{proof}


\section{Lipschitz experts in a (very) high dimension: regret characterization}
\label{sec:FFproblem-characterization}

As it turns out, the log-covering dimension~\refeq{eq:CovDim} is not the right notion to characterize regret for arbitrary metric spaces. We need a more refined version, similar to the \emph{max-min-covering dimension} from~\cite{LipschitzMAB-stoc08}:
\begin{align}
\MaxMinLCD(X) = \textstyle{\sup_{Y\subset X}} \;
    \inf\{ \,\LCD(Z):\; \text{open non-empty $Z\subset Y$} \}.
\end{align}
\noindent Note that in general $\MaxMinLCD \leq \LCD(X)$. The equality holds for ``homogenous" metric spaces such as \eps-uniform tree metrics. We prove that Theorems~\ref{thm:experts-MaxMinLCD} holds with $b = \MaxMinLCD(X)$.

\begin{theorem}\label{thm:experts-MaxMinLCD-body}
Fix a metric space $(X,d)$ and let $b=\MaxMinLCD(X)$. The \FFproblem\ on $(X,d)$ is $(t^\gamma)$-tractable for any $\gamma> \tfrac{b+1}{b+2}$, and not $(t^\gamma)$-tractable for any $\gamma < \tfrac{b-1}{b}$.
\end{theorem}

For the lower bound, we use a suitably ``thick'' version of the ball-tree from Section~\ref{sec:lower-bound} in conjunction with the $(\eps,\delta,k)$-ensemble
idea from Section~\ref{sec:lower-bound}, see Section~\ref{sec:MaxMinLCD-LB}. For the algorithmic result, we combine the ``naive" experts algorithm (\NaiveExp) with (an extension of) the \emph{transfinite fat decomposition} technique from~\cite{LipschitzMAB-stoc08}, see Section~\ref{sec:MaxMinLCD-UB}.

The lower bound in Theorem~\ref{thm:experts-MaxMinLCD-body} holds for the \ULproblem. It follows that the upper bound in Theorem~\ref{thm:ulproblem-naive} is optimal for metric spaces such that $\MaxMinLCD(X) = \LCD(X)$, e.g. for \eps-uniform tree metrics. In fact, we can plug the improved analysis of \NaiveExp\ from Theorem~\ref{thm:ulproblem-naive} into the algorithmic technique from Theorem~\ref{thm:experts-MaxMinLCD-body} and obtain a matching upper bound in terms of the \MaxMinLCD. Thus (in conjunction with Theorem~\ref{thm:main-experts}) we have a complete characterization for regret:

\OMIT{ 
The extension to the corresponding result for $\MaxMinLCD$ (the upper bound in Theorem~\ref{thm:experts-MaxMinLCD-unif}) proceeds exactly as in the proof of Theorem~\ref{thm:experts-MaxMinLCD}, except we use a more efficient analysis of \NaiveExp. We omit the details from this version.
} 

\begin{theorem}\label{thm:experts-MaxMinLCD-unif-body}
Consider the \ULproblem\ with full feedback. Fix a metric space $(X,d)$ with uncountably many points, and let $b=\MaxMinLCD(X)$. The problem on $(X,d)$ is $(t^\gamma)$-tractable for any
    $\gamma> \max(\tfrac{b-1}{b},\, \tfrac12)$,
and not $(t^\gamma)$-tractable for any
    $\gamma < \max(\tfrac{b-1}{b},\, \tfrac12)$.
\end{theorem}

The proof of the upper bound in Theorem~\ref{thm:experts-MaxMinLCD-unif-body} proceeds exactly that in Theorem~\ref{thm:experts-MaxMinLCD-body}, except that we use a more efficient analysis of \NaiveExp.

\subsection{The $\MaxMinLCD$ lower bound: proof for Theorem~\ref{thm:experts-MaxMinLCD-unif-body}}
\label{sec:MaxMinLCD-LB}

If $\MaxMinLCD(X) = d,$ and $\gamma < \tfrac{d-1}{d},$
let us first choose $b < c < d$ such that
$\gamma < \tfrac{b-1}{b}$.  Let $Y \subseteq X$
be a subspace such that $c \leq \inf \{\LCD(Z) : \mbox{open,
nonempty } Z \subseteq Y \}.$
We will repeatedly use the following packing lemma that
relies on the fact that $b < \LCD(U)$ for all
nonempty $U \subseteq Y$.

\begin{lemma} \label{lem:packing}
For any nonempty
open $U \subseteq Y$ and any $r_0>0$ there exists
$r \in (0, r_0)$ such that $U$ contains more than $2^{r^{-b}}$
disjoint balls of radius $r$.
\end{lemma}
\begin{proof}
Let $r < r_0$ be a positive number such that every
covering of $U$ requires more than $2^{r^{-b}}$ balls of
radius $2r$.  Such an $r$ exists, because
$\LCD(U) > b$.
Now let $\mathcal{P} = \{B_1,B_2,\ldots,B_M\}$ be any
maximal collection of disjoint $r$-balls.  For
every $y \in Y$ there must exist some ball $B_i
\; (1 \leq i \leq M)$ whose center is within
distance $2r$ of $y$, as otherwise $B(y,r)$
would be disjoint from every element of $\mathcal{P}$
contradicting the maximality of that collection.
If we enlarge each ball $B_i$ to a ball $B_i^+$
of radius $2r$, then every $y \in Y$ is contained
in one of the balls $\{B_i^+ \,|\, 1 \leq i \leq M\}$,
i.e. they form a covering of $Y$.  Hence
$M \geq 2^{r^{-b}}$ as desired.
\end{proof}

Using the packing lemma we recursively
construct an infinite sequence of sets
$\mathcal{B}_0, \mathcal{B}_1, \ldots$ each consisting of
finitely many disjoint open balls of equal
radius $r_i$ in $Y$.  Let $\mathcal{B}_0 = \{Y\}$
and let $r_0=1/4$.
If $i>0$, let $r_i < r_{i-1}/4$ be a positive number
small enough that for every  ball $B=B(x,r_{i-1}) \in \mathcal{B}_{i-1}$,
the sub-ball $B(x,r_{i-1}/2)$ contains
$n_i = \lceil 2^{r_i^{-b}} \rceil$ disjoint balls
of radius $r_i$.
Let $\mathcal{B}_i(B)$ denote this collection of
disjoint balls and let
$\mathcal{B}_i = \bigcup_{B \in \mathcal{B}_{i-1}}
\mathcal{B}_i(B).$
For each ball $B=B(x,r) \in \mathcal{B}_{i}$,
define a ``bump function'' supported in $B$ by
    $$ G_B(x) = \begin{cases}
         \min\{ r - d(x,s_i), r/2 \} & \mbox{if $x \in B$} \\
         0 & \mbox{otherwise}.
       \end{cases} $$
Let $\mathcal{B}_{\infty} = \cup_{i=0}^{\infty} \mathcal{B}_i$.
Note that $\mathcal{B}_{\infty}$ is analogous to the ball-tree defined in Section~\ref{sec:lower-bound}.

For a mapping $\sigma \,:\, \mathcal{B}_{\infty} \rightarrow
\{ \pm 1 \},$ define a function $\payoff \,:\, X \rightarrow
[0,1]$ by
\begin{equation} \label{eq:mmlcdlb-payoff}
\payoff(x) = \tfrac12 + \textstyle{\sum_{B \in \mathcal{B}_{\infty}}}
\sigma(B) G_B(x).
\end{equation}
The sum converges absolutely because every
$x \in X$ belongs to at most ball in $\mathcal{B}_i$
for each $i$, so the absolute value of the infinite
sum in \eqref{eq:mmlcdlb-payoff} is bounded
above by $\sum_{i=0}^{\infty} r_i < 1/3$.  Moreover,
one can verify that our construction ensures that
$\payoff(x)$ is a Lipschitz function of $x$ with
Lipschitz constant $1$.

If $\mathcal{Q}$ is any subset of $\mathcal{B}_{\infty}$,
one can define a problem instance of the \FFproblem: a distribution $\prob_{\mathcal{Q}}$
on payoff functions $\pi$ by sampling $\sigma(B) \in \{ \pm 1 \}$
uniformly at random for $B \not\in \mathcal{Q}$
and performing biased sampling of $\sigma(B) \in \{ \pm 1 \}$
with $E[\sigma(B)] = 1/3$ when $B \in \mathcal{Q}$,
and then defining $\pi$ using \eqref{eq:mmlcdlb-payoff}.
Note that the distribution $\prob_{\mathcal{Q}}$ has expected
payoff function $\mu = \tfrac12 + \sum_{B \in \mathcal{Q}} G_B/3.$

In proving the lower bound, we will consider the
distribution $\mathcal{P}$ on Lipschitz experts
problem instances $\prob_{\mathcal{Q}}$ where
$\mathcal{Q}$ is a random subset of
$\mathcal{B}_{\infty}$ obtained by
sampling one ball
$B_0 \in \mathcal{B}_0$
uniformly at random, and also sampling
one element $Q(B) \in \mathcal{B}_i(B)$
uniformly at random and independently
for each $B \in \mathcal{B}_{i-1}$.
By analogy with the notion of complete
lineage defined in Section~\ref{sec:lower-bound},
we will refer to any such set $\mathcal{Q}$
as a complete lineage in $\mathcal{B}_{\infty}$.
Given a complete lineage $\mathcal{Q}$,
we can define an infinite nested sequence
of balls $B_0 \supset B_1 \supset \cdots$
by specifying that $B_{i+1} = Q(B_i)$
for each $i$.

If $\mu$ is the expectation
of a random payoff function  $\pi$ sampled from $\prob_{\mathcal{Q}}$,
then $\mu$ achieves its maximum value
$\tfrac12 + \tfrac16 \sum_{i=0}^{\infty} r_i$
at the unique point $x^* \in \cap_{i=0}^{\infty} B_i$.
At any point $x \not\in B_j$, we have
\[
\mu(x^*) - \mu(x) \;
    \geq \; \textstyle{ \left(\tfrac16\, \sum_{i=j}^{\infty} r_i\right)}
  -  \textstyle{ \left( \tfrac14\, \sum_{i=j+1}^{\infty} r_i \right) \; = \; \tfrac16\, r_j.}
\]
We now finish the lower bound proof as in the proof
of Lemma~\ref{lm:ball-tree-LB}.
For each complete lineage $\mathcal{Q}$ and
ball $B \in \mathcal{B}_{i-1}$, let
$B^1,B^2,\ldots,B^{n_i}$ be the elements
of $\mathcal{B}_i(B)$.  Consider the
sets
   $\mathcal{Q}_0 = \mathcal{Q} \setminus Q(B)$
and
   $\mathcal{Q}_j = \mathcal{Q}_0 \cup \{B^j\}$
for $j=1,2,\ldots,n_i$.  The distributions
$\left(
\prob_{\mathcal{Q}_0},\prob_{\mathcal{Q}_1},\ldots,\prob_{\mathcal{Q}_{n_i}}
\right)$
constitute an $(\eps,\delta,k)$-ensemble
for $\eps=\tfrac{1}{6}r_k$, $\delta=\tfrac{1}{2},$
and $k=n_i$.
Consequently, for $t_i = r_i^{-b}$, the inequality
$t_i < \ln(17k)/2\delta^2$ holds, and we
obtain a lower bound of
  $$ R_{(\A, \, \prob_{\mathcal{Q}_j})}(t_i) > \eps\, t_i / 2 = \Omega( r_i^{1-b})
 = \Omega(t_i^{(b-1)/b}) $$
for at least half of the distributions $\prob_{\mathcal{Q}_j}$
in the ensemble.  Recalling that $\gamma < \tfrac{b-1}{b}$,
we see that the problem is not $t^{\gamma}$-tractable.

\subsection{The $\MaxMinLCD$ upper bound: proofs for
Theorem~\ref{thm:experts-MaxMinLCD-body} and Theorem~\ref{thm:experts-MaxMinLCD-unif-body}
}
\label{sec:MaxMinLCD-UB}

\newcommand{\NaiveSample}{\ensuremath{\text{{\sc NaiveSample}}}}

First, let us incorporate the analysis from Section~\ref{sec:FFproblem} via the following lemma.

\begin{lemma}\label{lm:ULproblem-recap}
Consider an instance $(X,d,\prob)$ of the \FFproblem, and let $x^*\in X$ be an optimal point. Fix subset $U\subset X$ which contains $x^*$, and let  $b>\LCD(U)$. Then for any sufficiently large $T$ and $\delta = T^{-1/(b+2)}$ the following holds:
\begin{itemize}
\item[(a)] Let $S$ be a $\delta$-hitting set for $U$ of cardinality $|S|\leq N_\delta(U)$. Consider the feedback of all points in $S$ over $T$ rounds; let $x$ be the point in $S$ with the largest sample average (break ties arbitrarily). Then
	$$ \Pr[\mu(x^*) - \mu(x) <O(\delta\log T)]> 1-T^{-2}.$$

\item[(b)] For a \ULproblem\ and $b\geq 2$, property (a) holds for $\delta = T^{-1/b}$.

\end{itemize}
\end{lemma}

\xhdr{Transfinite LCD decomposition.}
We redefine the \emph{transfinite fat decomposition} from~\cite{LipschitzMAB-stoc08} with respect to the log-covering dimension rather than the covering dimension.

\begin{definition}\label{def:fatness-transfinite}
Fix a metric space $(X,d)$. Let $\beta$ denote an arbitrary ordinal.
A \emph{transfinite LCD decomposition} of depth $\beta$ and dimension $b$ is a  transfinite sequence
	$\{S_\lambda\}_{0 \leq \lambda \leq \beta}$
of closed subsets of $X$ such that:
\begin{OneLiners}
\item[(a)] $S_0 = X$, $S_\beta = \emptyset$, and
	$S_\nu \supseteq S_\lambda$ whenever $\nu < \lambda$.
\item[(b)] if $V\subset X$ is closed, then the set
    $\{\text{ordinals } \nu \leq \beta$:\, $V \mbox{ intersects } S_\nu \}$
has a maximum element.
\item[(c)] for any ordinal $\lambda \leq \beta$ and any open set
$U\subset X$ containing $S_{\lambda+1}$ we have
	$\LCD(S_\lambda \setminus U) \leq b$.
\end{OneLiners}
\end{definition}

The existence of suitable decompositions and the connection to $\MaxMinLCD$ is derived exactly as in Proposition 3.15 in~\cite{LipschitzMAB-stoc08}.

\begin{lemma} \label{prop:fatness-dim}
For every compact metric space $(X,d)$, $\MaxMinLCD(X)$ is equal to the infimum of all $b$ such that $X$ has a transfinite LCD decomposition of dimension $b$.
\end{lemma}

\newcommand{\DpthOracle}{{\mathtt{Depth}}}
\newcommand{\DCovOracle}{{\mathtt{Cover}}}
\newcommand{\ve}{\varepsilon}
\newcommand{\PhaseAlg}{\ensuremath{\A_{\mathtt{ph}}}}
\newcommand{\rank}{\ensuremath{\mathtt{depth}}}


In what follows, let us fix metric space $(X,d)$ and $b>\MaxMinLCD(X)$, and let
	$\{S_\lambda\}_{0 \leq \lambda \leq \beta}$
be a transfinite LCD decomposition of depth $\beta$ and dimension $b$. For each $x\in X$, let the \emph{depth} of $x$ be the maximal ordinal $\lambda$ such that $x\in S_\lambda$. (Such an ordinal exists by Definition~\ref{def:fatness-transfinite}(b).)

\xhdr{Access to the metric space.}
The algorithm requires two oracles: the \emph{depth oracle} $\DpthOracle(\cdot)$ and the \emph{covering oracle} $\DCovOracle(\cdot)$. Both oracles input a finite collection $\F$  of open balls $B_0, B_1, \ldots, B_n$, given via the centers and the radii, and return a point in $X$. Let $B$ be the union of these balls, and let $\overline{B}$ be the closure of $B$. A call to oracle $\DpthOracle(\F)$ returns an arbitrary point $x\in \overline{B} \cap S_\lambda$, where $\lambda$ is the maximum ordinal such that $S_{\lambda}$ intersects $\overline{B}$. (Such an ordinal exists by Definition~\ref{def:fatness-transfinite}(b).) Given a point $y^*\in X$ of depth $\lambda$, a call to oracle $\DCovOracle(y^*,\F)$ either reports that $B$ covers $S_{\lambda}$, or it returns an arbitrary point $x \in S_{\lambda} \setminus B$. A call to $\DCovOracle(\emptyset, \F)$ is equivalent to the call $\DCovOracle(y^*,\F)$ for some $y^*\in S_0$.

The covering oracle will be used to construct $\delta$-nets as follows. First, using successive calls to $\DCovOracle(\emptyset, \F)$ one can construct a $\delta$-net for $X$. Second, given a point $y^*\in X$ of depth $\lambda$ and a collection of open balls whose union is $B$, using successive calls to $\DCovOracle(y^*,\,\cdot)$ one can construct a $\delta$-net for $S_\lambda \setminus B$. The second usage is geared towards the scenario when $S_{\lambda+1} \subseteq B$ and  for some optimal strategy $x^*$ we have
	$x^*\in S_\lambda \setminus B$.
Then by Definition~\ref{def:fatness-transfinite}(c) we have
	$\LCD(S_\lambda\setminus B)<b$,
and one can apply Lemma~\ref{lm:ULproblem-recap}.

\xhdr{The algorithm.}
Our algorithm proceeds in phases $i=1,2,3,\ldots$ of $2^i$ rounds each. Each phase $i$ outputs two strategies:  $x^*_i, y^*_i\in X$ that we call the \emph{best guess} and the \emph{depth estimate}. Throughout phase $i$, the algorithm plays the best guess $x^*_{i-1}$ from the previous phase. The depth estimate $y^*_{i-1}$ is used ``as if" its depth is equal to the depth of some optimal strategy. (We show that for a large enough $i$ this is indeed the case with a very high probability.)

In the end of the phase, an algorithm selects a finite set $A_i\subset X$ of \emph{active points}, as described below. Once this set is chosen, $x^*_i$  is defined simply as a point in $A_i$ with the largest sample average of the feedback (breaking ties arbitrarily). It remains to define $y^*_i$ and $A_i$ itself.

Let $T=2^i$ be the phase duration. Using the covering oracle, the algorithm constructs (roughly) the finest $r$-net containing at most $2^{\sqrt{T}}$ points. Specifically, the algorithm constructs  $2^{-j}$-nets $\mathcal{N}_j$, for $j = 0,1,2,\ldots$, until it finds the largest $j$ such that
	$\mathcal{N}_j$
contains at most $2^{\sqrt{T}}$  points. Let
	$r = 2^{-j}$ and $\mathcal{N} = \mathcal{N}_j$.

For each $x\in X$, let $\mu_T(x)$ be the sample average of the feedback during this phase. Let
\begin{align*}
 \Delta_T(x) &= \mu^*_T - \mu_T(x),
        \text{~~~where~~~}
        \mu^*_T = \max(\mu_T, \mathcal{N})
\end{align*}

\noindent Define the depth estimate $y^*_i$ to be the output of the oracle call $\DpthOracle(\F)$, where
$$ \F = \{ B(x,r):\; x\in \mathcal{N} \text{~~and~~} \Delta_T(x)< r \}.
$$

Finally, let us specify $A_i$. Let $B$ be the union of balls
\begin{align}\label{eq:ffproblem-PMO-B}
\{ B(x,r):\; x\in \mathcal{N} \text{~~and~~} \Delta_T(x)> 2(r_T + r)\, \},
\end{align}
where
            $r_T = \sqrt{8\log(T\,|\mathcal{N}|)/T}$
is chosen so that by Chernoff Bounds we have
\begin{align}\label{eq:ffproblem-PMO-Chernoff}
 \Pr[ | \mu_T(x) - \mu(x)| <  r_T ] > 1-  (T \,|\mathcal{N}|)^{-3}
    \quad\text{for each $x\in \mathcal{N}$}.
\end{align}

\noindent Let  $\delta = T^{-1/b}$ for the \ULproblem, and
	$\delta = T^{-1/(b+2)}$
otherwise. Let
    $Q_T =  2^{\delta^{-b}} $
be the \emph{quota} on the number of active points. Given a point $y^*_{i-1}$ whose depth is (say) $\lambda$, algorithm uses the covering oracle to construct a $\delta$-net $\mathcal{N'}$ for $S_\lambda \setminus B$. Define $A_i$ as $\mathcal{N}'$ or an arbitrary $Q_T$-point subset thereof, whichever is smaller.\footnote{The interesting case here is $|\mathcal{N}'| \leq Q_T$.  If $\mathcal{N}'$ contains too many points, the choice of $A_i$ is not essential for the analysis.}

\xhdr{(Very high-level) sketch of the analysis.}
The proof roughly follows that of Theorem 3.16 in~\cite{LipschitzMAB-stoc08}.
Call a phase \emph{clean} if the event in~\refeq{eq:ffproblem-PMO-Chernoff} holds for all $x\in \mathcal{N}_i$ and the appropriate version of this event holds for all $x\in A_i$. (The regret from phases which are not clean is negligible). On a very high level, the proof consists of two steps. First we show that for a sufficiently large $i$, if phase $i$ is clean then the depth estimate $y^*_i$ is correct, in the sense that it is indeed equal to the depth of some optimal strategy. The argument is similar to the one in Lemma~\ref{lm:tractability}. Second, we show that for a sufficiently large $i$, if the depth estimate $y^*_{i-1}$ is ``correct" (i.e. its depth is equal to that of some optimal strategy), and phase $i$ is clean, then the ``best guess" $x^*_i$ is good, namely $\mu(x^*_i)$ is within $O(\delta log T)$ of the optimum. The reason is that, letting $\lambda$ be the depth of $y^*_{i-1}$, one can show that for a sufficiently large $T$ the set $B$ (defined in~\refeq{eq:ffproblem-PMO-B}) contains	$S_{\lambda+1}$ and does not contain some optimal strategy. By definition of the transfinite LCD decomposition we have
    $\LCD(S_\lambda \setminus U) < b$,
so in our construction the quota $Q_T$ on the number of active points permits $A_i$ to be a $\delta$-cover of $S_\lambda \setminus U$. Now we can use Lemma~\ref{lm:ULproblem-recap} to guarantee the ``quality" of $x^*_i$. The final regret computation is similar to the one in the proof of Theorem~\ref{thm:ffproblem-naive}.

\begin{small}
\bibliographystyle{plain}
\bibliography{bib-abbrv,bib-bandits,bib-math,bib-extras}
\end{small}

\appendix

\section{Reduction to compact metric spaces}
\label{sec:reduction}

In this section we reduce the Lipschitz MAB problem to that on complete metric spaces.

\begin{lemma}\label{lm:reduction}
The Lipschitz MAB problem on a metric space $(X,d)$ is $f(t)$-tractable if and only if it is $f(t)$-tractable on the completion of $(X,d)$. Likewise for the Lipschitz experts problem with double feedback.
\end{lemma}

\begin{proof}
Let $(X,d)$ be a metric space with completion $(Y,d)$. Since $Y$ contain an isometric copy of $X$, we will abuse notation and consider $X$ as a subset of $Y$. We will present the proof the Lipschitz MAB problem; for the experts problem with double feedback, the proof is similar.

Given an algorithm $\A_X$ which is $f(t)$-tractable for $(X,d)$, we may use it as a Lipschitz MAB algorithm for $(Y,d)$ as well.  (The algorithm has the property that it never selects a point of $Y \setminus X$, but this doesn't prevent us from using it when the metric space is $(Y,d)$.)  The fact that $X$ is dense in $Y$ implies that for every Lipschitz payoff function $\mu$ defined on $Y$, we have
    $ \sup(\mu,X) = \sup(\mu,Y). $
From this, it follows immediately that the regret of $\A_X$, when considered a Lipschitz MAB algorithm for $(X,d)$, is the same as its regret when considered as a Lipschitz MAB algorithm for $(Y,d)$.

Conversely, given an algorithm $\A_Y$ which is $f(t)$-tractable for $(Y,d)$, we may design a Lipschitz MAB algorithm $\A_X$ for $(X,d)$ by running $\A_Y$ and perturbing its output slightly. Specifically, for each point $y\in Y$ and each $t\in \N$ we fix $x = x(y,t) \in X$ such that $d(x, y)<2^{-t}$. If $\A_Y$ recommends playing strategy $y_t \in Y$ at time $t$, algorithm $\A_X$ instead plays $x = x(y,t)$. Let $\pi$ be the observed payoff. Algorithm $\A_X$ draws an independent 0-1 random sample with expectation $\pi$, and reports this sample to $\A_Y$. This completes the description of the modified algorithm $\A_X$.

Suppose $\A_X$ is not $f(t)$-tractable. Then for some problem instance $\mathcal{I}$ on $(Y,d)$, letting $R_X(t)$ be the expected regret of $\A_X$ on this instance, we have that
    $\sup_{t\in\N} R_X(t)/f(t) =\infty$.
Let $\mu$ be the expected payoff function in $\mathcal{I}$.  Consider the following two problem instances of a MAB problem on $Y$, called $\mathcal{I}_1$ and $\mathcal{I}_2$, in which if point $y\in Y$ is played at time $t$, the payoff is an independent 0-1 random sample with expectation $\mu(y)$ and $\mu( x(y,t))$, respectively. Note that algorithm $\A_Y$ is $f(t)$-tractable on $\mathcal{I}_1$, and its behavior on $\mathcal{I}_2$ is identical to that of $\A_X$ on the original problem instance $\mathcal{I}$. It follows that by observing the payoffs of $\A_Y$ one can tell apart $\mathcal{I}_1$ and $\mathcal{I}_2$ with high probability. Specifically, there is a ``classifier" $\mathcal{C}$ which queries one point in each round, such that for infinitely many times $t$ it tell apart $\mathcal{I}_1$ and $\mathcal{I}_2$ with success probability $p(t) \to 1$. Now, the latter is information-theoretically impossible.

To see this, let $H_t$ be the $t$-round history of the algorithm (the sequence of points queried, and outputs received), and consider the distribution of $H_t$ under problem instances $\mathcal{I_1}$ and $\mathcal{I_2}$ (call these distributions $q_1$ and $q_2$). Let us consider  and look at their KL-divergence. By the chain rule (See Lemma~\ref{lem:kl}), we can show that $KL(q_1,q_2) <\tfrac12$. (We omit the details.) It follows that letting $S_t$ be the event that $\mathcal{C}$ classifies the instance as $\mathcal{I}_1$ after round $t$, we have
    $\mathbb{P}_{q_1}[S_t] - \mathbb{P}_{q_2}[S_t] \leq KL(q_1,q_2) \leq \tfrac12$.
For any large enough time $t$,
    $\mathbb{P}_{q_1}[S_t] <\tfrac14$,
in which case $\mathcal{C}$ makes a mistake (on $\mathcal{I}_2$) with constant probability.
\end{proof}

\OMIT{ 
If $\A_Y$ recommends playing strategy $y_t \in Y$ at time $t$, we instead find a point $x_t \in X$ such that $d(x_t,y_t) < 2^{-t}$ and we play $x_t$.  The regret of algorithm $\A_x$ exceeds that of $\A_Y$ by at most $\frac12 + \frac14 + \frac18 + \ldots = 1$, hence $\A_X$ is $f(t)$-tractable.
} 

\begin{lemma}\label{lm:reduction-2}
Consider The experts MAB problem with full feedback. If it is $f(t)$-tractable on a metric space $(X,d)$ then it is $f(t)$-tractable on the completion of $(X,d)$.
\end{lemma}

\begin{proof}
Identical to the easy (``only if") direction of Lemma~\ref{lm:reduction}.
\end{proof}

\begin{note}{Remark.}
Lower bounds only require Lemma~\ref{lm:reduction-2}, or the easy (``only if") direction of Lemma~\ref{lm:reduction}. For the upper bounds (algorithmic results), we can either quote the ``if" direction of Lemma~\ref{lm:reduction}, or prove the desired property directly for the specific type algorithms that we use (which is much easier but less elegant).
\end{note}


\OMIT{
For any $B \in \mathcal{B}_{\infty}, \, t \in \N$,
let $\sigma(B,t)$ denote
the value of $\sigma(B)$ sampled at time $t$
when sampling the sequence of i.i.d.~payoff
functions $\payoff_t$ from distribution $\prob_{\mathcal{Q}}$.
We know that $\sigma(B_k,t)=1$ for all $k,t$.
In fact if $S(k,t)$ denotes the set of all
$B \in \mathcal{B}_k(B_{k-1})$ such that
    $\sigma(B,1) = \sigma(B,2) = \cdots = \sigma(B,t) = 1$
then conditional on the value of the set $S(k,t)$, the
value of $B_k$ is distributed uniformly at random
in $S(k,t)$.  As long as this set $S(k,t)$
has at least $n$ elements, the probability
that the algorithm picks a strategy $x_t$ belonging
to $B_k$ at time $t$ is bounded above by $\tfrac{1}{n}$,
even if we condition on the event that
    $x_t \in B_{k-1}.$
For any given $B \in \mathcal{B}_k(B_{k-1}) \setminus \{B_k\}$, we have
    $\prob_{\mathcal{Q}}(B \in S(k,t)) = 2^{-t}$
and these events are independent for different
values of $B$.  Setting $n=\sqrt{n_k}$, so that
$|\mathcal{B}_k(B_{k-1})|=n^2$, we have
    \begin{align}
      \nonumber
    \prob_{\mathcal{Q}} \left( |S(k,t)| \leq n \right) &\leq
    \sum_{R \subset \mathcal{B}_k(B_{k-1}), |R|=n} \prob_{\mathcal{Q}}(S(k,t)
\subseteq R) \\
      \nonumber
    &= \binom{n^2}{n}
    \left( 1-2^{-t} \right)^{n^2-n} \\
      \nonumber
    &< \left( n^2 \cdot \left( 1 - 2^{-t} \right)^{n-1} \right)^{n} \\
      \label{eq:mmlcdlb-3}
    &< \exp \left( n (2 \ln(n) - (n-1)/2^t) \right).
    \end{align}
STILL NEED TO FIX THIS.
Let $t_k = \tfrac14 \log_2(n_k) = \tfrac14 r_k^{-b}$
As long as $t < t_k$, the relation $\log_2(n) = \tfrac12
\log_2(n_k)$ implies $(n-1)/2^t > \sqrt{n}$ so the expression
\eqref{eq:mmlcdlb-3} is bounded above by
$\exp \left( -n \sqrt{n} + 2 n \ln(n) \right)$,
which equals
$\exp \left( -n_k^{3/4} + n_k^{1/2} \ln(n_k) \right)$
and is in turn bounded above by $\exp \left( -\tfrac12 n_k^{3/4} \right).$

For any Lipschitz experts algorithm \A,
let $N(t,k)$ denote the random variable
that counts the number of times \A\ selects
a strategy in $B_k$ during rounds $1,\ldots,t$.
We have already demonstrated that for all $t \leq t_k$,
    \begin{equation} \label{eq:mmlcdlb-4}
       \Pr_{\prob_\mathcal{Q} \in \mathcal{P}}(x_t \in B_k)
       \leq n_k^{-1/2} + \exp \left( -\tfrac12 n_{k}^{3/4} \right)
       < 2 n_k^{-1/2},
    \end{equation}
where the term $n_k^{-1/2}$ accounts for the
event that $S(k,t)$ has at least $\sqrt{n_k}$ elements.
Equation \eqref{eq:mmlcdlb-4} implies the bound
    $\mathbb{E}_{\prob_{\mathcal{Q}} \in \mathcal{P}}[N(t_k,k)] < 2 t_k n_k^{-1/2}.$
By Markov's inequality, the probability that $N(t_k,k) > t_k/2$
is less than $4 n_k^{-1/2}$.
By Borel-Cantelli, almost
surely the number of $k$ such that $N(t_k,k) \leq t_k/2$
is finite.
The algorithm's expected regret at time $t$ is
bounded below by $r_k(t_k - N(t_k,k))$, so with
probability $1$, for all but finitely many $k$ we have
  $$R_{(\A, \, \prob_\mathcal{Q})}(t_k) \geq r_k t_k / 2
    = r_k^{1-b}$
This establishes that \A\ is not $g(t)$-tractable.
}

\section{KL-divergence techniques}
\label{sec:KL-divergence}

Our proof will use the notion of Kullback-Leibler
divergence (or \emph{KL-divergence}), defined for
two probability measures as follows.

\begin{definition}  \label{def:kldiv}
Let $\Omega$ be a finite set with
two probability measures $p,q$.  Their \emph{Kullback-Leibler
divergence}, or KL-divergence, is the sum
\[
KL(p;q) = \sum_{x \in \Omega} p(x) \ln \left(
\frac{p(x)}{q(x)} \right),
\]
with the convention that $p(x) \ln(p(x)/q(x))$
is interpreted to be $0$ when $p(x)=0$ and $+\infty$
when $p(x)>0$ and $q(x)=0$.  If $Y$ is a random
variable defined on $\Omega$ and taking values in
some set $\Gamma$, the \emph{conditional
Kullback-Leibler divergence} of $p$ and $q$
given $Y$ is the sum
\[
KL(p;q \given Y) = \sum_{x \in \Omega} p(x)
\ln \left( \frac{p(x \given Y = Y(x))}{q(x \given Y = Y(x))} \right),
\]
where terms containing $\log(0)$ or $\log(\infty)$ are
handled according to the same convention as above.
\end{definition}

The definition can be applied to an infinite sample
space $\Omega$ provided that $q$ is absolutely
continuous with respect to $p$.  For details,
see~\cite{Bobby-thesis}, Chapter 2.7.
The following lemma summarizes some standard facts about
KL-divergence; for proofs, see~\cite{CoverThomas,Bobby-thesis}.
\begin{lemma} \label{lem:kl}
Let  $p,q$ be two probability measures on
a measure space $(\Omega,\mathcal{F})$ and
let $Y$ be a random variable defined on
$\Omega$ and taking values in some finite set
$\Gamma$.  Define a
pair of probability measures $p_Y,q_Y$ on $\Gamma$ by
specifying that $p_Y(y) = p(Y=y), q_Y(y) = q(Y=y)$ for
each $y \in \Gamma$.  Then
\[
KL(p;q) = KL(p;q \given Y) + KL(p_Y;q_Y),
\]
and $KL(p;q \given Y)$ is non-negative.
\end{lemma}
An easy corollary is the following lemma
which expresses the KL-divergence of two
distributions on sequences as a sum of
conditional KL-divergences.
\begin{lemma} \label{lem:kl-chainrule}
Let $\Omega$ be a sample space, and suppose $p,q$
are two probability measures on $\Omega^n$,
the set of $n$-tuples of elements of $\Omega$.
For a sample point $\vec{\omega} \in \Omega^n$,
let $\omega^i$ denote its  first $i$ components.
If $p^i,q^i$ denote the probability
measures induced on $\Omega^i$ by $p$
(resp. $q$) then
\[
KL(p;q) = \textstyle{\sum_{i=1}^n}\, KL(p^i;q^i \given \omega^{i-1}).
\]
\end{lemma}
\begin{proof}  For $m=1,2,\ldots,n$, the formula
$KL(p^m;q^m) = \sum_{i=1}^m KL(p^i;q^i \given \omega^{i-1})$
follows by induction on $m$, using Lemma~\ref{lem:kl}.
\end{proof}
The following three lemmas will also be useful
in our lower bound argument.  Here and henceforth
we will use the following notational convention:
for real numbers $a,b \in [0,1]$,
$KL(a;b)$ denotes the KL-divergence
$KL(p;q)$ where $p,q$ are probability
measures on $\{0,1\}$ such that $p(\{1\})=a,
\, q(\{1\}) = b.$  In other words,
\[
KL(a;b) = a \ln \left( \tfrac{a}{b} \right) +
(1-a) \ln \left( \tfrac{1-a}{1-b} \right).
\]
\begin{lemma} \label{lem:kl-bernoulli}
For any $0 < \eps < y \leq 1$,
$KL(y-\eps;y) < \eps^2/y(1-y).$
\end{lemma}
\begin{proof}
A calculation using the inequality $\ln(1+x)<x$
(valid for $x > 0$) yields
\begin{align*}
KL(y-\eps;y) &= (y-\eps) \ln \left( \tfrac{y-\eps}{y} \right)
+ (1-y+\eps) \ln \left( \tfrac{1-y+\eps}{1-y} \right) \\
&< (y-\eps) \left( \tfrac{y-\eps}{y} - 1 \right)
+ (1-y+\eps) \left( \tfrac{1-y+\eps}{1-y} - 1 \right) \\
& = \tfrac{-\eps(y-\eps)}{y} + \tfrac{\eps(1-y+\eps)}{1-y}
 = \tfrac{\eps^2}{y(1-y)}.\qedhere
\end{align*}
\end{proof}

\begin{lemma} \label{lem:kl-distinguishing}
Let $\Omega$ be a sample space with
two probability measures
$p,q$ whose KL-divergence is $\kappa.$  For
any event $\mathcal{E}$, the probabilities
$p(\mathcal{E}), \, q(\mathcal{E})$ satisfy
\[
q(\mathcal{E}) \geq
p(\mathcal{E}) \exp \left(
- \tfrac{\kappa + 1/e}{p(\mathcal{E})} \right).
\]
\end{lemma}
A consequence of the lemma, stated in less quantitative
terms, is the following: if $\kappa=KL(p;q)$ is bounded above
and $p(\mathcal{E})$ is bounded away from zero then
$q(\mathcal{E})$ is bounded away from zero.
\begin{proof}
Let $a = p(\mathcal{E}), \, b = q(\mathcal{E}),
c = (1-a)/(1-b)$.
Applying Lemma~\ref{lem:kl} with $Y$ as the indicator
random variable of $\mathcal{E}$ we obtain
\[
\kappa = KL(p;q) \geq KL(p_Y;q_Y) =
a \ln \left( \tfrac{a}{b} \right) +
(1-a) \ln \left( \tfrac{1-a}{1-b} \right) =
a \ln \left( \tfrac{a}{b} \right) + (1-b)\, c \ln(c).
\]
Now using the inequality $c \ln(c) \geq -1/e$,
(valid for all $c \geq 0$) we obtain
\[
\kappa \geq a \ln(a/b) - (1-b)/e \geq a \ln(a/b) - 1/e.
\]
The lemma follows by rearranging terms.
\end{proof}
\begin{lemma} \label{lem:reverse-pinsker}
Let $p,q$ be two probability measures, and suppose that
for some $\delta \in (0,\tfrac{1}{2}]$ they satisfy
\[
\forall \mbox{\rm \ events } \mathcal{E}, \quad
1-\delta < \tfrac{q(\mathcal{E})}{p(\mathcal{E})}
< 1+\delta
\]
Then $KL(p;q) < \delta^2.$
\end{lemma}
\begin{proof}
We will prove the lemma assuming the sample space
is finite.  The result for general measure spaces
follows by taking a supremum.

For every $x$ in the sample space $\Omega$, let
$r(x) = \frac{q(x)}{p(x)}-1$ and note that $|r(x)| < \delta$
for all $x$.  Now we make use of the inequality
$\ln(1+x) \leq x-x^2$, valid for $x \geq -\tfrac{1}{2}.$
\begin{align*}
KL(p;q) &= \textstyle{\sum_{x}}\, p(x) \ln \left( \tfrac{p(x)}{q(x)} \right)  \quad\quad\;
 = \textstyle{\sum_{x}}\, p(x) \ln \left( \tfrac{1}{1 + r(x)} \right) \\
&= - \textstyle{\sum_{x}}\, p(x) \ln(1+r(x))
\;\leq - \textstyle{\sum_{x}}\, p(x) [r(x) - (r(x))^2 ] \\
& < - \left( \textstyle{\sum_{x}}\, p(x) r(x) \right) +
    \delta^2 \left( \textstyle{\sum_{x}}\, p(x) \right)
\\
& = - \left( \textstyle{\sum_{x}}\, q(x) - p(x) \right) + \delta^2
 = \delta^2. \qedhere
\end{align*}
\end{proof}

\begin{proofof}{Theorem~\ref{thm:LB-technique}}
Let $\Omega = [0,1]^X$.  Using Property~\ref{ens:1}
of an $(\eps,\delta,k)$-ensemble combined with
Lemma~\ref{lem:reverse-pinsker}, we find that
$KL(\prob_i;\prob_0) < \delta^2.$

Let \A\ be an experts algorithm whose random bits
are drawn from a sample space $\Gamma$
with probability measure $\nu$.  For any
positive integer $s < \ln(17k) / 2 \delta^2$,
let $p_i^s$ denote
the measure $\nu \times (\prob_i)^s$
on the probability space $\Gamma \times \Omega^s.$
By the chain rule for KL-divergence
(Lemma~\ref{lem:kl-chainrule}),
$KL(p_i^s;p_0^s) < s \delta^2 <  \ln(17k) / 2.$
Now let $\mathcal{E}_i^s$ denote the event that
\A\ selects a point $x \in S_i$ at time $s$.
If $p_i^s(\mathcal{E}_i^s) \geq \tfrac{1}{2}$
then Lemma~\ref{lem:kl-distinguishing} implies
\begin{align*}
p_0^s(\mathcal{E}_i^s)
& \geq
p_i^s(\mathcal{E}_i^s) \exp \left(
- \frac{\ln(17k)/2 + 1/e}{p_i^s(\mathcal{E}_i^s)}
\right)
 \geq
\tfrac{1}{2} \exp \left(- \ln(k) + \ln(17) - \tfrac{2}{e} \right)
 > \frac{4}{k}.
\end{align*}
The events $\{\mathcal{E}_i^s \,|\, 1 \leq i \leq k\}$
are mutually exclusive, so fewer than $k/4$ of them
can satisfy $p_0^s(\mathcal{E}_i^s) > \frac{4}{k}.$
Consequently, fewer than $k/4$ of them can satisfy
$p_i^s(\mathcal{E}_i^s) \geq \tfrac{1}{2},$ a
property we denote in this proof by saying that
$s$ is \emph{satisfactory} for $i$.
Now assume $t < \ln(17k)/2 \delta^2$.
For a uniformly random $i \in \{1,\ldots,k\}$,
the expected number of satisfactory
$s \in \{1,\ldots,t\}$
is less than $t/4$, so by Markov's inequality, for
at least half of the $i \in \{1,\ldots,k\}$, the
number of satisfactory $s \in \{1,\ldots,t\}$ is
less than $t/2$.  Property~\ref{ens:2} of an
$(\eps,\delta,k)$-ensemble guarantees that
every unsatisfactory $s$ contributes at least
$\eps$ to the regret of \A\ when the problem
instance is $\prob_i$.  Therefore, at least half
of the measures $\prob_i$ have the property that
	$R_{(\A,\,\prob_i)}(t) \geq \eps t/2$.
\end{proofof}

\subsection{Proof of Claim~\ref{cl:logT-KLdiv}}
\label{sec:logT-KLdiv}

Recall that in Section~\ref{sec:logT} we defined
a pair of payoff functions $\mu_0,\mu_i$ and
a ball $B_i$ of radius $r_i$ such that $\mu_0 \equiv \mu_i$
on $X \setminus B_i$, while for $x \in B_i$ we have
    $$ \tfrac38 \leq \mu_0(x) \leq \mu_i(x) \leq
       \mu_0(x) + \tfrac{r_i}{4} \leq \tfrac34.
    $$
Thus, by Lemma~\ref{lem:kl-bernoulli},
$KL(\mu_0(x);\mu_i(x)) < r_i^2 / 3$ for
all $x \in X$, and $KL(\mu_0(x);\mu_i(x)) = 0$
for $x \not\in B_i$.

Represent the algorithm's choice and the payoff
observed at any given time $t$ by a pair $(x_t,y_t).$
Let $\Omega = X \times [0,1]$ denote the set of
all such pairs.  When a given algorithm \A\
plays against payoff functions $\mu_0, \mu_i$,
this defines two different probability measures
$p_0^t, p_i^t$ on the set $\Omega^t$ of possible
$t$-step histories.  Let $\omega^t$ denote a
sample point in $\Omega^t$.  The bounds derived
in the previous paragraph imply that for any
non-negative integer $s$,
\begin{equation} \label{eq:logT-KLdiv-1}
KL(p_0^{s+1}; p_i^{s+1} \,|\, \omega^s) <
\tfrac{1}{3} r_i^2 \prob_0(x_{s+1} \in B_i).
\end{equation}
Summing equation \eqref{eq:logT-KLdiv-1} for
$s=0,1,\ldots,t-1$ and applying Lemma~\ref{lem:kl-chainrule}
we obtain
\begin{equation} \label{eq:logT-KLdiv-2}
KL(p_0^{t}; p_i^{t}) < \tfrac13 r_i^2 \;
    \textstyle{\sum_{s=1}^{t}}\, \prob_0(x_{s} \in B_i)
 = \tfrac13 r_i^2 \mathbb{E}_0(N_i(t)),
\end{equation}
where the last equation follows from the definition
of $N_i(t)$ as the number of times algorithm \A\
selects a strategy in $B_i$ during the first $t$ rounds.

The bound stated in Claim~\ref{cl:logT-KLdiv}
now follows by applying Lemma~\ref{lem:kl-distinguishing}
with the event $S$ playing the role of $\mathcal{E}$,
$\prob_0$ playing the role
of $p$, and $\prob_i$ playing the role of $q$.

\section{Topological equivalences: proof of Lemma~\ref{lm:topological-equivalence}}
\label{sec:topological}

Let us restate the lemma, for the sake of convenience. Recall that it includes an equivalence result for compact metric spaces, and two implications for arbitrary metric spaces:

\begin{lemma}\label{lm:topological-equivalence-appendix}
For any compact metric space $(X,d)$, the following are equivalent: (i) $X$ is a countable set, (ii) $(X,d)$ is well-orderable, (iii) no subspace of $(X,d)$ is perfect. For an arbitrary metric space we have (ii)$\iff$(iii) and (i)$\Rightarrow$(ii), but not (ii)$\Rightarrow$(i).
\end{lemma}

\begin{proof}[Proof: compact metric spaces]
Let us prove the assertions in the circular order.

{\bf (i) implies (iii).} Let us prove the contrapositive: if $(X,d)$ has a
perfect subspace $Y$, then $X$ is uncountable. We have
seen that if $(X,d)$ has a perfect subspace $Y$ then it has
a ball-tree. Every leaf $\ell$ of the ball-tree (i.e. infinite path
starting from the root) corresponds to a nested sequence
of balls. The closures of these balls
have the finite intersection property, hence their
their intersection is non-empty. Pick an arbitrary point of
the intersection and call if $x(\ell)$. Distinct leaves $\ell$, $\ell'$
correspond to distinct points $x(\ell)$, $x(\ell')$ because if
$(y,r_y), \, (z,r_z)$ are siblings in the ball-tree
which are ancestors of $\ell$ and $\ell'$, respectively,
then the closures of $B(y,r_y)$ and $B(z,r_z)$ are disjoint
and they contain $x(\ell), x(\ell')$ respectively.
Thus we have constructed
a set of distinct points of $X$, one for each leaf of the
ball-tree. There are uncountably many leaves, so $X$ is
uncountable.

{\bf (iii) implies (ii).}  Let $\beta$ be some ordinal of strictly larger cardinality than $X$. Let us define a transfinite sequence
	$\{x_\lambda\}_{\lambda\leq \beta}$
of points in $X$ using transfinite recursion\footnote{''Transfinite recursion" is a theorem in set theory which asserts that in order to define a function $F$ on ordinals, it suffices to specify, for each ordinal $\lambda$, how to determine $F(\lambda)$ from $F(\nu)$, $\nu<\lambda$.}, by specifying that $x_0$ is any isolated point of $X$, and that for any ordinal $\lambda > 0$, $x_\lambda$ is any isolated point of the subspace $(Y_\lambda, d)$, where
	$Y_\lambda = X\setminus \{ x_\nu:\, \nu < \lambda \}$,
as long as $Y_\lambda$ is nonempty. (Such isolated point exists since by our assumption subspace $(Y_\lambda, d)$ is not perfect.) If $Y_\lambda$ is empty define e.g. $x_\lambda = x_0$.
Now, $Y_\lambda$ is empty for some ordinal $\lambda$ because otherwise we obtain a mapping from $X$ onto an ordinal $\beta$ whose cardinality exceeds the cardinality of $X$. Let
	$\beta_0 = \min\{ \lambda:\, Y_\lambda=\emptyset \}$.
Then every point in $X$ has been indexed by an ordinal number $\lambda<\beta_0$, and so we obtain a well-ordering of $X$. By construction, for
every $x=x_\lambda$ we can define  a radius $r(x)>0$ such that $B(x, r(x))$ is disjoint from the set of points $\{ x_\nu : \nu > \lambda \}$.
Any initial segment $S$ of the well-ordering is equal to
the union of the balls $\{B(x,r(x)) : x \in S\}$, hence is an
open set in the metric topology. Thus we have constructed
a topological well-ordering of X.

{\bf (ii) implies (i).} Suppose we have a binary relation $\prec$
which is a topological well-ordering of $(X,d)$. Let $S(n)$ denote
the set of all $x \in X$ such that $B(x, \tfrac{1}{n})$ is contained in
the set $P(x) = \{ y : y \preceq x \}$. By the definition of a
topological well-ordering we know that for every $x$, $P(x)$ is
an open set, hence $x\in S(n)$ for sufficiently
large $n$. Therefore $X = \cup_{n\in\N} S(n)$. Now, the definition
of $S(n)$ implies that every two points of $S(n)$ are
separated by a distance of at least $1/n$. (If $x$ and $z$
are distinct points of $S(n)$ and $x\prec z$, then $B(x,\tfrac{1}{n})$
is contained in the set $P(x)$ which does not contain
$z$, hence $d(x,z)\geq \tfrac{1}{n}$.) Thus by compactness of $(X,d)$ set $S(n)$ is finite.
\end{proof}

\begin{proof}[Proof: arbitrary metric spaces]
For implications {\em (i)$\Rightarrow$(ii)}  and {\em (iii)$\Rightarrow$(ii)}, the proof above does not in fact use compactness. An example of an uncountable but well-orderable metric space is $(\R,d)$, where $d$ is a uniform metric. It remains to prove that {\em (ii)$\Rightarrow$(iii)}.

Suppose there exists a topological well-ordering $\prec$. For each subset $Y\subseteq X$ and an element $\lambda\in Y$ let
	$Y_\prec(\lambda) = \{ y\in Y: y\preceq \lambda\}$
be the corresponding initial segment.

We claim that $\prec$ induces a topological well-ordering on any subset $Y\subseteq X$. We need to show that for any $\lambda\in Y$ the initial segment $Y_\prec(\lambda)$ is open in the metric topology of $(Y,d)$. Indeed, fix $y\in Y_\prec(\lambda)$. The initial segment $X_\prec(\lambda)$ is open by the topological well-ordering property of $X$, so
	$B_X(y,\eps) \subset X_\prec(\lambda)$
for some $\eps>0$. Since
	$Y_\prec(\lambda) =  X_\prec(\lambda) \cap Y$
and
	$B_Y(y,\eps) = B_X(y,\eps) \cap Y $,
it follows that
	$B_Y(y,\eps) \subset Y_\prec(\lambda)$.
Claim proved.

Suppose the metric space $(X,d)$ has a perfect subspace $Y\subset X$. Let $\lambda$ be the $\prec$-minimum element of $Y$. Then $Y_\prec(\lambda) = \{\lambda \}$. However, by the previous claim $\prec$ is a topological well-ordering of $(Y,d)$, so the initial segment $Y_\prec(\lambda)$ is open in the metric topology of $(Y,d)$. Since $(Y,d)$ is perfect, $Y_\prec(\lambda)$ must be infinite, contradiction. This completes the {\em (ii)$\Rightarrow$(iii)} direction.
\end{proof}

\end{document}